\documentclass{article}

\usepackage[paperwidth=210mm,paperheight=297mm,centering,hmargin=2.5cm,vmargin=2.5cm]{geometry}
\usepackage{pgfplots}
\usepackage{graphicx}
\usepackage{dcolumn}
\usepackage{bm}
\usepackage[colorlinks= true, linkcolor=blue, citecolor=blue, urlcolor=blue]{hyperref}
\usepackage{lipsum}
\usepackage{caption}
\usepackage{subcaption}
\usepackage{amsthm}
\usepackage{amssymb}
\usepackage{amsfonts}
\usepackage{array,multirow,graphicx}
\usepackage{float}
\usepackage{amsmath}
\usetikzlibrary{calc,arrows.meta,bending}
\usepackage{xparse}
\usepackage{xargs} 
\usepackage{tikz}
\usepackage{todonotes}
\usepackage[framemethod=tikz]{mdframed}
\usepackage{lineno}
\newmdenv[leftmargin=\dimexpr-0.4em, innerleftmargin=0.5em,
          rightmargin=\dimexpr-0.4em, innerrightmargin=0.5em,
          linewidth=2pt,linecolor=red, topline=false, bottomline=false,
          innertopmargin=0pt,innerbottommargin=0pt,skipbelow=0pt,skipabove=0pt,%
         ]{notex}

\newenvironment{note}%
 {\vskip\dimexpr\dp\strutbox-\prevdepth\relax\notex\strut\ignorespaces}%
 {\xdef\notetpd{\the\prevdepth}\endnotex\vskip-\notetpd\relax}

\let\oldtodo\todo

\makeatletter%
\DeclareDocumentCommand{\todo}{ O{} +g +d<> }{%
    \IfNoValueTF{#2}{\relax}{%
         \oldtodo[caption={#2},size=\footnotesize,#1]{\renewcommand{\baselinestretch}{1}\selectfont\sffamily#2\par}%
    }%
    \IfNoValueTF{#3}{\relax}{%
        \IfNoValueTF{#2}{
            \begin{note}%
                \begin{internallinenumbers}%
                    \indent%
                    #3%
                \end{internallinenumbers}%
            \end{note}%
    }{
        \vspace{-0\baselineskip}%
        \begin{note}%
            \begin{internallinenumbers}%
                \indent%
                #3%
            \end{internallinenumbers}%
        \end{note}%
        }%
    }%
}%
\makeatother
\usepackage{soul}

\makeatletter
\newcommand*{\rom}[1]{\expandafter\@slowromancap\romannumeral #1@}
\makeatother

\pgfplotsset{compat=1.16}
\usepackage{authblk}
\begin{document}

\newtheorem{remark}{Remark}
\newtheorem{theorem}{Theorem}
\newtheorem{corollary}{Corollary}
\newtheorem{lemma}{Lemma}
\newtheorem{prop}{Proposition}
\newtheorem{defn}{Definition}
\newtheorem{assume}{Assumption}
\newcounter{example}
\newenvironment{example}[1][]{\refstepcounter{example}\par\medskip
   \noindent \textbf{Example~\theexample. #1} \rmfamily}{\medskip}

\title{\textbf{A Scalable Bayesian Persuasion Framework for Epidemic Containment on Heterogeneous Networks}}

\author[1]{Shraddha Pathak} 
\author[2]{Ankur A. Kulkarni\thanks{Corresponding author. This research was supported by the grant
CRG/2019/002975 of the Science and Engineering
Research Board, Department of Science and Technology,
India.}}
\affil[1]{Indian Institute of Science Education and Research Pune, Pune 411008, India. Email: \texttt{shraddhapathak1999@gmail.com}}
\affil[2]{Indian Institute of Technology Bombay, Mumbai, 400076, India. Email: \texttt{kulkarni.ankur@iitb.ac.in}}
\renewcommand\Affilfont{\itshape\small}
\date{}

\maketitle

 \begin{abstract}
 During an epidemic, the information available to individuals in the society deeply influences their belief of the epidemic spread, and consequently the preventive measures they take to stay safe from the infection. In this paper, we develop a scalable framework for ascertaining the optimal information disclosure a government must make to individuals in a networked society for the purpose of epidemic containment. 
 This problem of information design problem is complicated by the heterogeneous nature of the society, the positive externalities faced by individuals, and the variety in the public response to such disclosures. We use a networked public goods model to capture the underlying societal structure. Our first main result is a structural decomposition of the government's objectives into two independent components -- a component dependent on the utility function of individuals, and another dependent on properties of the underlying network. Since the network dependent term in this decomposition is unaffected by the signals sent by the government, this characterization simplifies the problem of finding the optimal information disclosure policies. We find explicit conditions, in terms of the risk aversion and prudence, under which no disclosure, full disclosure, exaggeration and downplay are the optimal policies. The structural decomposition results are also helpful in studying other forms of interventions like incentive design and network design.

 \end{abstract}

\section{Introduction} \label{sec:intro}
The recent COVID-19 pandemic challenged the global public health system in unforeseeable ways. Governments across the world formed various policies, such as mandating testing before travel, enforcing lockdown, fining non-mask wearers and so on, as attempts to control the spread of the pandemic. Mobile applications with varied purposes such as contact tracing and health monitoring \cite{mobileapps} were also developed in many countries. A universal feature of these applications is their ability to share information about the epidemic with the public. The prevalence of smartphones and real-time information gathering and sharing technologies, opens the possibility for the government to strategically control the information environment of the public and induce pro-social actions by changing their perception of the epidemic spread. Inspired by this novel aspect of intervention design to control the epidemic, in this paper, we \textit{study the efficacy of signals sent by the government in achieving social objectives and determine the optimal information disclosure policy for epidemic containment.}

Designing interventions and policies, in general, is not a trivial task, especially given the complex structure of the society and the multiple factors which influence the decisions of individuals, such as their perception of the epidemic spread and the actions of other individuals. Individuals in the society are not homogeneously connected. Different individuals have different social circles and meet different number of people. Moreover, the precautionary measures taken by individuals to stay safe during epidemics depend on their beliefs of the infection level in the society. Individuals tend to be more careful when they think the epidemic is more prevalent, while they tend to relax their guard when they believe that there are not many infections in the population. Precautionary actions taken by individuals also bestow benefits (perhaps, unintended) or positive externalities on their contacts since these actions reduce the contacts' likelihood of getting infected. These externalities lead to strategic interactions (with a substitutes effect) among the actions of neighbouring individuals. For example, knowing that one's friend does not meet many people, one might relax one's guard and not take the necessary precautions while meeting this friend. These effects complicate the problem of intervention design since incentivizing an individual to be more careful can have a cascading detrimental effect on neighbours' actions, possibly leading to an overall negative effect of the incentive.

While accounting for the heterogeneous nature of the society and the strategic decisions of individuals, governments devise policies and intervention plans to contain the epidemic spread. Unlike the 1918 Spanish flu pandemic, the advances in technologies have equipped us with innovative ways to handle infectious diseases. In this paper, we study how signaling policies can be designed to send information about the epidemic, which change individuals' actions by changing their beliefs of the prevalence of the epidemic, thereby inducing certain social outcomes. 
Another complication arises from the fact that the government usually has more than one (mathematically different but heuristically similar) objectives that it would want the society to attain.

In this paper, we develop a framework for ascertaining the optimal information disclosure by a government to agents (common people) that are connected by a network and face positive externalities. The payoff of each agent is the difference between its benefit, i.e. the probability of being safe, and the cost that it incurs. The benefit depends on the efforts (such as wearing masks) made by the agent and its neighbours, but the cost are only associated with the agent's own effort. These positive externalities result in a public goods game among agents situated on an underlying network and having incomplete information about the infection level in the society, where each agent has to choose its optimal effort level. The signals sent by the government affect the beliefs of these agents, and government crafts these signals to maximize its objectives. We study the optimal signals for two possible objectives of the government: maximizing society's total effort to stay safe during the epidemic, and maximizing the probability that a randomly selected individual is safe. For each signal, the resulting public goods game between agents has multiple Nash equilibria. Thus we consider two possible attitudes for computing the government's objective: an optimistic one which considers the best case over all Nash equilibria, and a pessimistic one which considers the worst case.

Our first main finding is the structural decomposition of the above complicated objectives of the government into two independent components: a network dependent term and a term that depends only on the utility function of the agents. 
Since the graph-dependent term is not affected by informative signals, such reductions are valuable as they allow us to circumvent the complexities that arise due to the network, and design signals affecting effectively unilateral 
individual actions. These reductions are also helpful while studying other forms of intervention design, as we discuss later in the paper.

Using the above result, we find conditions under which full disclosure, no disclosure, exaggeration, and downplaying are the optimal policies. These conditions are in terms of the risk aversion and prudence of the benefit function. Furthermore, we find sufficient conditions, which do not depend on the beliefs of individuals, under which the extreme policies of full information and no information are optimal. This opens up the possibility of inferring the optimal policies solely based on the nature of the probability of being safe in different infection states, without requiring details about the beliefs of individuals in the society. Remarkably, we also find that the solutions of the optimal signaling policy are robust under the two different objectives and the two attitudes of the government, to a large extent. Moreover, the structural characterization also makes our results for optimal public signaling policies robust to different network structures, and thus scalable to large networks.

\subsection{Related works} 

For a given belief of the epidemic state, the game played by individuals in the society falls in the realm of network games \cite{jackson, bramoullegames2015}. The two broad categories of games on network are games of strategic complements (where actions of players mutually reinforce one another) and games of strategic substitutes (actions taken by players disincentivize others from taking the same actions). As discussed before, the epidemic game is a game of strategic substitutes and we model the protection received from the precautionary actions of individuals as a public good \cite{bramoulle} -- giving public, non-excludable benefits (from staying safe during the epidemic) but coming at individual costs (of efforts).

On top of the epidemic game played by individuals in the society at a given belief of the epidemic state, by using the framework of Bayesian persuasion and information design \cite{persuasion}, we allow for the government to strategically send signals about the epidemic state to induce different beliefs among individuals in the public, and to thereby achieve certain social outcomes by changing the underlying epidemic game (modelled as a public goods game) being played by individuals in the society. Bayesian persuasion is classically studied as a two-player game where one player (called the sender) decides on a signaling policy according to which information about a payoff-relevant state is sent to the other player (called the receiver). Upon receiving the signal, the receiver chooses an action, which determines the payoff for both, the receiver as well as the sender. Multi-receiver extension of this two-player game, often referred to as information design, has also received wide interest \cite{bergemann2019information}. Even though multi-receiver settings with inter-dependent actions are receiving some attention now \cite{multiagentleverage}, these settings are less well understood because of the complexities associated with them. Our work gives insights on information design in such settings. 

The frameworks of Bayesian persuasion and information design have found applications in a wide range of settings like vehicle routing \cite{das2017reducing}, voting \cite{voters} and grade disclosures \cite{gradedisclosure} and our work is an application in the important problem of epidemic control. Despite the tremendous amount of literature on studying the role of human decision making on the spread of epidemics \cite{epidemicgamereview}, not enough attention has been paid to designing effective intervention strategies. While \cite{coronagames} studies the effect of direct incentives on the decision of individuals, \cite{informingpublicpandemic} and \cite{eproach} study effect of information on the behaviour of individuals. 

The work closest to ours is Ref. \cite{informingpublicpandemic} which also poses the problem of informing the public about a pandemic in the framework of Bayesian persuasion. However, our work differs from theirs in two crucial ways. First, we consider a continuous action space, in contrast to the discrete (binary) action space considered in \cite{informingpublicpandemic}. Furthermore, one of our main focus is on the effect of the network of connections among individuals on the design of interventions, whereas \cite{informingpublicpandemic} assumes a homogeneously connected population.

\subsection{Organization of the paper} We begin by discussing relevant literature on the public goods game and Bayesian persuasion in Section \ref{sec:preliminaries}, which together equip us with tools to formulate and analyse our problem. In Section \ref{sec:model}, we formulate our model and discuss its salient features. Section \ref{sec:results} contains our main results and we conclude the paper in Section \ref{sec:discussion} by discussing the implications of our analysis on policy design. Appendix \ref{sec:Analysis} presents the proofs of our results.

\section{Preliminaries and notation} \label{sec:preliminaries}

The preliminaries discussed in this section are divided into two components. We begin by reviewing the public goods game on networks \cite{bramoulle, parthe}, which captures the interactions among individuals in the society during an epidemic crisis. In the second half of this section, we discuss the framework of Bayesian persuasion \cite{persuasion}, which allows us to account for the role of signals sent by the government in our problem.

\subsection{Public goods game on networks} \label{subs:publicgood}

Bramoull{\'e} and Kranton \cite{bramoulle} introduced a public goods game on a network. Public goods are goods which involve benefits which are non-excludable. The costs of actions taken for provision of the goods, however, are incurred at the individual level, leading to strategic interactions among individuals. The strategic substitutes nature of the game allows for individuals to free-ride on the actions of others, while obtaining full benefit from the provision of the good.

In the model described by Bramoull{\'e} and Kranton in \cite{bramoulle}, individuals are situated on the nodes of a graph $G(V,E)$ and are connected to other individuals whose actions give them direct benefits. An individual $k\in V$ exerts effort $x_k\in [0,\infty)$ for the provision of the public good. The collective neighbourhood effort (for individual $k$) $\mathcal{E}_k(x) := x_k + \sum_{j\in N_k} x_j$ yields a benefit according to a strictly concave increasing function $b(\cdot)$, while the individual effort comes at a marginal cost $c$. Thus, the utility of individual $k$ in the public goods game is given by equation \eqref{eq:public goods} below.
\begin{align}
    U_k(x) = b(x_k + \sum_{j\in N_k} x_j) - c\cdot x_k = b(\mathcal{E}_k(x)) - c\cdot x_k \label{eq:public goods}
\end{align}

Rational individuals playing this game exert effort which maximizes their utility given by \eqref{eq:public goods} above. Let $ e^*$ be the effort level where marginal benefit equals marginal cost, i.e., \begin{align}
     b'(e^*)=c. \label{eq:net_eff}
\end{align} This effort level $e^*$ is what we call the `unilateral effort'. This is the effort exerted by individuals in the absence of network externalities. Given the neighbourhood effort $\sum_{j\in N_k}x_j$, Bramoull{\'e} and Kranton \cite{bramoulle} show that the best response strategy for player $k$ is $x_k = e^* - \sum_{j \in N_k}x_j$, as long as $\sum_{j\in N_k}x_j\leq e^*$. Otherwise, $x_k=0$.
Further, they characterized the Nash equilibria\footnote[1]{A Nash equilibrium is a solution concept for the entire system of individuals in which every individual is playing its best response, and thus \textit{no} individual has an incentive to deviate.} \cite{nash1950equilibrium} of the game into three categories. A \textit{specialized equilibrium} is when individuals in the society either exert full effort $e^*$, or exert no effort and free-ride on their neighbourhood effort i.e., $x_k\in \{0,e^*\}$, $\forall k\in V$. In a \textit{distributed equilibrium} every individuals exerts some positive effort i.e., $x_k >0$ $\forall k\in V$. A \textit{hybrid equilibrium} is one which is neither specialized nor distributed i.e., and some individuals free-ride while some others exert effort less than $e^*$. 

Remarkably, Bramoull{\'e} and Kranton \cite{bramoulle} map every specialized equilibrium of this game to a maximal independent set\footnote[2]{An \textit{independent set} $S$ of a graph is a subset of the vertices of the graph ($S\subset V$) which are all mutually non-adjacent i.e., no two nodes in $S$ are neighbours of each other. A \textit{maximal independent set} is an independent set which is not strictly a subset of any other independent set of the graph i.e., for any node $k\in S^c$, $S\cup\{k\}$ is no longer an independent set.} $S$ of the underlying network $G$. In this mapping, every individual in the independent set exerts effort $e^*$, while the ones not in the independent set free-ride. This is because every node in the independent set is not connected to any other element of the independent set (by definition) and every node not in the independent set is connected to at least one node in the independent set (since if it is not, it can be added to the original set to form an even larger independent set, contradicting the fact that it is maximal) leading to $\sum_{j\in N_k}x_j =0$, when $k\in S$ and $\sum_{j\in N_k} x_j\geq e^*$, for $k \notin S$.

Since the Nash equilibria effort profiles of the public goods game might not lead to the best welfare among all possible effort profiles, Pandit and Kulkarni \cite{parthe} ask how the different types of equilibria (specialized, distributed and hybrid) perform among themselves. Remarkably, they show that among all Nash equilibria, the specialized equilibrium corresponding to the $w$-weighted maximum independent set\footnote[3]{For weights given by the vector $w=(w_1,...,w_n)$, where $w_k$ corresponds to the weight of the $k^{th}$ node, a \textit{$w-$weighted maximum independent set} is the maximal independent set $S$ maximizing $\alpha_w(G)=\sum_{k\in S} w_k$.} attains the maximum $w$-weighted aggregate effort of the society ($\sum_{k\in V} w_k\cdot x_k$). Moreover, they show that distributed equilibria, whenever they exist, lead to the minimum aggregate cost ($c\cdot\sum_{k\in V} x_k$) among all Nash equilibria. These results are the building blocks for our analysis and have been discussed in detail in Lemma \ref{lem:parthe}.

\subsection{Bayesian persuasion and the information design framework} \label{subs:BayPer}
We now take a detour from the realm of public goods games and instead focus our study on a particular kind of game of asymmetric information called Bayesian persuasion \cite{persuasion}. In this game, there are two players -- the sender and the receiver. There is a payoff-relevant state which is unknown to the receiver. However, both players have a common belief (modelled as probability distributions over the set of the states of the world) on the realised value of the state. The sender commits to a signaling policy $\pi(\cdot|\cdot)$ (which is a conditional distribution $\pi(\cdot|r)$ on the signalled states of the world, given the real state of the world $r$) according to which signals about the state are sent to the receiver. The receiver updates its belief based on the signal received and the committed policy using Bayes' rule, and takes an action that maximizes its own expected utility. The receiver's action determines the payoff of the receiver as well as the sender. The sender's action is to choose the policy which maximizes its own expected payoff.

When the uncertain state of the world and the signals sent are both binary taking values in the set $\{A, B\}$, finding the optimal signaling policy $\pi(s| r)$ ($s$ stands for the signalled state of the world) is equivalent to finding two real numbers $p_a, p_b\in [0,1]$ where $p_a= \pi(s = A|r=A)$ is the probability with which $A$ is communicated when the real state is $A$ and $p_b= \pi(s = B|r=B)$ is the probability with which $B$ is communicated when the real state is $B$. In this system, a \textit{full information disclosure policy} corresponds to the situation where all the information available is truthfully communicated to the receiver i.e.,
\begin{align}
    \text{Full disclosure:}\hspace{2em} \pi(s=A|r=A)=1, \hspace{1em} \pi(s=B|r=B)=1. \label{eq:fulldisc}
\end{align}
Equivalently, full disclosure corresponds to the case when $(p_a,p_b)=(1,1)$. The case where the sender always lies about the state $(p_a, p_b)=(0,0)$ also fully discloses the real state since the receiver can correctly infer the real state (because it knows that the sender always lies). A \textit{no information disclosure policy}, on the other hand, corresponds to the case where the probability with which state $A$ is communicated is the same irrespective of the real state i.e., 
\begin{align}
    \text{No disclosure:}\hspace{2em}\pi(s=A|r=A) = \pi(s=A|r=B). \label{eq:nodisc}
\end{align}
This is all cases where $p_a+p_b=1$. In this, the posterior belief generated by the signal is the same as the prior belief. There are many policies in between these two extremes where partial information can be communicated to the receiver.

Kamenica and Gentzow \cite{persuasion} show that the nature of the optimal signaling policy depends on the curvature (concavity/convexity) of the sender's utility as a function of the receiver's beliefs. When the sender's utility is strictly concave in the receiver's belief, providing no information is optimal. On the other hand, when the sender's utility is strictly convex in the receiver's belief, giving full information is optimal. In the intermediate cases where the sender's utility is neither strictly concave nor convex in the receiver's beliefs, an intermediate disclosure policy maybe optimal. These results are proved using a convex analysis argument involving the concave closure of sender's utility as a function of the receiver's belief. The concave closure $\mathcal{C}_{f}$ of a function $f : [a,b] \rightarrow \mathbb{R}$ is defined as
\begin{align}
    C_f(x) := \inf \{ \ell(x) | \ell \in L_f \} \qquad \forall x \in [a,b] \label{eq:concave closure}
\end{align} where $L_f := \{\ell | \ell \text{ is concave and } \ell \geq f \text{ over } [a,b]\}$. For any function $f$, $C_f$ is a concave function.

\section{Model} \label{sec:model}
\subsection{Setting} \label{sec:setting}

We now briefly outline the precise setting we consider. We consider a society of individuals that is affected by an epidemic. Individuals interact socially through a heterogeneous but fixed number of interactions and take precautionary measures such as wearing masks, sanitizing, and vaccination to avoid getting infected. These actions can be quantified in terms of the effort individuals take to stay safe during the epidemic. 
Individuals benefit from their own effort and the effort taken by their neighbours, since it increases their probability of being safe. 
This causes the benefit obtained from the effort to be non-excludable, though the costs are incurred by the individual. Thus public health is a networked public good \cite{bramoulle} and the amount of effort an individual makes is a strategic decision.

The probability that an individual is safe during the epidemic and the benefit thereof, is also dependent on the infection level in the population.
Individuals do not have the exact knowledge of the infection level, but they have a belief over what it may be. Therefore, individuals choose their preventive measures by maximizing their expected payoff calculated with respect to the distribution given by this belief. Knowing the above situation, the government tries to induce socially optimal outcomes among the individuals of the society. It does so by leveraging the uncertainty that individuals have about the infection level and sending signals about it so as to alter their beliefs, and shapes these signals so as to maximize its own objective.

In the next section, we formalize this setting by writing down a mathematical model for it.

\subsection{Epidemic game as a network game}
Let the network of individuals in the society be given by the graph $G(V,E)$, where the individuals correspond to the vertices $V$ of the graph and the edges $E$ represent the interactions among the individuals. Let the effort taken by individual $k\in V$ to stay safe during the epidemic be $x_k\in [0,\infty)$. The infection level in the population is $\iota\in I=\{h,l\}$ (high or low), which individuals in the society are unaware about. The prior belief of individuals about the infection level is given by a probability distribution $\mu_0$ over $I$ i.e., individuals believe it to be a high infection state with probability $\mu_0(h)$ and a low infection state with probability $\mu_0(l) = 1-\mu_0(h)$.

An individual benefits from the cumulative effort made by other individuals in its neighbourhood. Let $$\mathcal{E}_k(x) := x_k + \sum_{j\in N_k} x_j$$
denote this cumulative effort in the neighbourhood of individual $k\in V$. 
The individual $k$ derives a benefit $b(\mathcal{E}_k(x); \iota)$ from staying safe when the infection level is $\iota$. We interpret the benefit function as the ``probability of being safe'' in the epidemic. As such, the benefit function $b$ is increasing and strictly concave in effort and is such that the benefit in the high infection state is strictly lesser than the benefit in the low infection state for the same effort level i.e., $b(x;h)<b(x;l)$, $\forall x \in [0, \infty)$.
The marginal cost associated with one's own effort is $c$, whereby an effort $x_k$ incurs a cost $c\cdot x_k$. When the infection level is known to be $\iota$, the utility of an individual is a deterministic quantity given by 
\begin{align*}
    U_k(x; \iota) = b(\mathcal{E}_k(x); \iota) - c\cdot x_k.
\end{align*}
However, since individuals in the society are uncertain about the infection level in the population, they choose actions which maximize their expected utility, which is given by
\begin{align*}
    \mathbb{E}_{\iota\sim \mu} [U_k(x; \iota)] = \mu(h)\cdot b(\mathcal{E}_k(x); h) + \mu(l)\cdot b(\mathcal{E}_k(x); l) - c\cdot x_k,
\end{align*}
where $\mu$ is the belief that individuals have about the infection level. 

In the absence of any signal from the government, we have $\mu=\mu_0$. However, the government provides information by communicating an infection level $\iota_S$.
The government's signaling policy $\pi$ is a conditional distribution of the signalled infection level $\iota_S\in I$ given the real infection level in the population $\iota$ i.e., $\iota_S \sim \pi(\cdot|\iota)$. We assume the individuals are Bayesian rationals, whereby the  belief about the infection level is updated according to Bayes' rule after receiving the government's signal, and is given by
\begin{align*}
    \mu_\pi(\iota|\iota_S) = \frac{\pi(\iota_S|\iota)\mu_0(\iota)}{\pi(\iota_S|\iota)\mu_0(\iota)+\pi(\iota_S|\iota^c)\mu_0(\iota^c)}.
\end{align*}
On receiving a signal $\iota_S$ from the government, individual $k$'s expected utility is 
\begin{align}
    \mathbb{E}_{\iota\sim \mu_\pi(\cdot|\iota_S)}[U_k(x; \iota) | \iota_S] &= \mathbb{E}_{\iota\sim \mu_\pi(\cdot|\iota_S)}[b(\mathcal{E}_k(x); \iota) - c\cdot x_k] \nonumber \\
    &= \mu_\pi(\iota=h|\iota_S)b(\mathcal{E}_k(x); h)+\mu_\pi(\iota=l|\iota_S)b(\mathcal{E}_k(x); l) -c\cdot x_k. \nonumber \\
    &= \tilde{b}(\mathcal{E}_k(x);\mu_\pi,\iota_S) - c\cdot x_k, \label{eq:epidemic game}
\end{align}
where $\tilde{b}(\mathcal{E}_k(x); \mu_\pi, \iota_S) := \mu_\pi(\iota=h|\iota_S)b((\mathcal{E}_k(x))_{\iota_S}; h)+\mu_\pi(\iota=l|\iota_S)b((\mathcal{E}_k(x))_{\iota_S}; l)$.
Thus, given a signal $\iota_S$, individual $k$ chooses effort level $x_k$ that maximizes this expected utility i.e., 
\begin{align}
    \textit{Individual's problem: }\hspace{2em}x_k = {\arg\max}_{x_k\in [0,\infty)} \mathbb{E}_{\iota\sim \mu_\pi(\cdot|\iota_S)}[U_k(x; \iota)].
\end{align}

Comparing \eqref{eq:public goods} and \eqref{eq:epidemic game}, we see that for every belief distribution $\mu_\pi$ and for every signal $\iota_S$, a distinct public goods game is played, with the benefit function as $\tilde{b}(\cdot; \mu_\pi, \iota_S)$. The monotonicity and concavity of $\tilde{b}$ in effort $x$ follows from the fact that convex combination of increasing and concave functions is also increasing and concave. 


\def\NE{{\rm NE}}
The expected utilities defined in \eqref{eq:epidemic game} results in a game at the level of the individuals in the society. The resulting Nash equilibrium  is defined as follows.
\begin{defn}[Nash equilibrium for the epidemic game]
Effort profile $x^*=(x^*_1,...,x^*_n)$ is a Nash equilibrium of the epidemic game played by the individuals in the society on receiving signal $\iota_S$ if
\begin{align*}
    x^*_k = x^*_k(\iota_S;\pi)&= \arg \max_{x_k} \mathbb{E}_{\iota\sim\mu_\pi(\cdot|\iota_S)}[U_k(x; \iota)|\iota_S], \hspace{1em} \forall k\in V\end{align*}
    where, $\mathbb{E}_{\iota \sim \mu_\pi(\cdot|\iota_S)}[U_k(x;\iota)| \iota_S]$ is given by \eqref{eq:epidemic game}.
    Let $\NE(\iota_S;\pi)$ denote the set of all Nash equilibria.
    \label{def:NE}
\end{defn}

We consider two aggregate-level objectives (denoted by $\mathcal{O}(x^*; \pi)$) 
that the government may have during an epidemic. The first objective is to maximize the expected total effort that all the individuals in the society together exert to stay safe during the epidemic. This is given by
    \begin{align}
        \mathcal{O}(x^*; \pi)=\mathbb{E}_{\iota\sim \mu_0}\mathbb{E}_{\iota_S\sim \pi(\cdot|\iota)}\bigg[\sum_{k\in V} x^*_k(\iota_S;\pi)\bigg]. \label{eq:obj1}
    \end{align}
The second objective we consider is to maximize the probability that a randomly selected individual is safe. Since the benefit function $b(;\iota)$ in an individual's utility equation \eqref{eq:epidemic game} denotes the probability with which it is safe during the epidemic, $\frac{1}{n}$ times the sum over every individual's benefit can be interpreted as the probability of a random individual being safe. This objective is thus given by
    \begin{align}
        \mathcal{O}(x^*; \pi)
        &= \frac{1}{n}\mathbb{E}_{\iota\sim \mu_0}\mathbb{E}_{\iota_S\sim \pi(\cdot|\iota)} \bigg[\sum_{k\in V}b(\mathcal{E}_k (x^*(\iota_S;\pi)); \iota)\bigg].\label{eq:obj2}
    \end{align}
From now on, we will represent the above double expectations in the government's objectives (in equations \ref{eq:obj1} and \ref{eq:obj2}) as $\mathbb{E}_{\iota, \iota_S}[\cdot]$.

Moreover, since the epidemic game at the level of the public can lead to multiple Nash equilibria, each performing differently in terms of the objectives, we consider two attitudes towards decision making: optimistic and pessimistic. An optimistic government looks for a signaling policy when the best Nash equilibrium is attained. A pessimistic government, on the other hand, is one which finds the best signaling policy when the worst possible Nash equilibrium is attained.

The solution for the government's problem of finding the optimal policy $\pi$ is given by the Stackelberg equilibrium \cite{stackelberg} of the game consisting of the decisions of the individuals as well as the government, as defined below.
\begin{defn}[Optimistic Stackelberg equilibrium for government-society incomplete information game]

$(\pi^*, x^*)$ is the optimistic Stackelberg equilibrium of the combined game between the individuals in the society and the government if
\begin{align*}
    \pi^* & \in \arg\max_\pi \max_{x^*\in \NE(\cdot;\pi)} \mathcal{O}(x^*; \pi),
\end{align*}
where $x^*\in \NE(\iota_S;\pi^*)$ satisfies definition \ref{def:NE}. \label{def:OSE}
\end{defn}

\begin{defn}[Pessimistic Stackelberg equilibrium for government-society incomplete information game]

$(\pi^*, x^*)$ is the pessimistic Stackelberg equilibrium of the combined game between the individuals in the society and the government if
\begin{align*}
    \pi^* &\in \arg\max_\pi \min_{x^*\in \NE(\cdot;\pi)} \mathcal{O}(x^*; \pi),
\end{align*}
where $x^*\in \NE(\iota_S;\pi^*)$ satisfies definition \ref{def:NE}. \label{def:PSE}
\end{defn}

\section{Results} \label{sec:results}

In this section, we present our two main results. The first result is a structural decomposition of the government's objectives into two independent components, one containing network related terms and another containing terms dependent on unilateral 
effort (as defined in \eqref{eq:net_eff}). The second set of results are related to government's optimal information disclosure policy. These results are obtained for every government objective and every government attitude. While the proofs of the results are presented in the appendix \ref{sec:Analysis}, we present their brief outline alongside the results here. 

Our analysis and results are based on the following two assumptions:

\begin{assume}
The unilateral effort of individuals $e^*$ \eqref{eq:net_eff} increases as the belief of the high infection state $\mu$ increases. \label{asu:1}
\end{assume}

This is a natural assumption since we expect one to take more effort to stay safe when one believes the infection level is high with a higher probability. 

\begin{assume}
The low infection level is infectious enough to ensure that individuals take some positive unilateral action, even when they are certain that the infection level is low (i.e., $c<b'(0;l)$).  \label{asu:2}
\end{assume}

While the analysis can be easily extended to relax Assumption \ref{asu:2}, the results presented are for cases when this assumption holds. However, we present an illustrative example for a case when this assumption is relaxed.

\subsection{Structural characterisation of the government's objectives} \label{subsec:res_struc}

\begin{theorem} (Structural characterization of government objectives)
\begin{enumerate}
    \item (Optimistic government maximizing aggregate effort)
    \begin{align*}
        \max_{x^*\in \NE} \mathbb{E}_{\iota, \iota_S}\bigg[\sum_{k \in V} x_k^*\bigg] = \alpha(G)\cdot \mathbb{E}_{\iota,\iota_S}[e^*]
    \end{align*}
    where $\alpha(G)$ is the size of the maximum independent set of the graph. Specialized equilibrium corresponding to the maximum independent set attains this equilibrium.
    \item (Pessimistic government maximizing aggregate effort) \begin{align*}
        \min_{x^*\in \NE}\mathbb{E}_{\iota,\iota_S}\bigg[\sum_{k \in V} x^*_k\bigg] = m(G)\cdot \mathbb{E}_{\iota,\iota_S}[e^*]
    \end{align*} where $m(G)$ is a constant that depends only on the graph $G$. Whenever distributed equilibria exist, they attain this aggregate effort.
    \item (Pessimistic government maximizing probability of random individual being safe) \begin{align*}
        \min_{x^*\in \NE}\mathbb{E}_{\iota,\iota_S}\bigg[\mathbb{P}(\text{a random node is safe})\bigg] = \mathbb{E}_{\iota,\iota_S}[b(e^*_{\iota_S};\iota)]
    \end{align*} Whenever distributed equilibria exist, they attain this value of the objective.
    \item (Optimistic government maximizing probability of random individual being safe) 
    
    Let $\sigma_b:= \frac{b(ne^*)-b(e^*)}{ce^*(n-1)}$.
\begin{enumerate} 
    \item As $\sigma_b\to 0$,  $$\max_{x^*\in \NE}\mathbb{E}_{\iota,\iota_S}\bigg[\mathbb{P}(\text{a random node is safe})\bigg]=\mathbb{E}_{\iota, \iota_S}[b(e^*_{\iota_S}; \iota)].$$ All Nash equilibria result in the above expected objective.
    \item As $\sigma_b\to 1$,  $$\max_{x^*\in \NE}\mathbb{E}_{\iota,\iota_S}\bigg[\mathbb{P}(\text{a random node is safe})\bigg]=\mathbb{E}_{\iota, \iota_S} [b(e^*_{\iota_S}; \iota) + f(G)\cdot e^*_{\iota_S}],$$ where $f(G)$ is a graph dependent constant. Specialised equilibrium corresponding to the $(A_G+I_n)e$-weighted maximum independent set\footnote[4]{Recall, for weights given by the vector $w=(w_1,...,w_n)$, where $w_k$ corresponds to the weight of the $k^{th}$ node, a \textit{$w-$weighted maximum independent set} is the maximal independent set $S$ maximizing $\alpha_w(G)=\sum_{k\in S} w_k$.} attains this value of the objective, where $e$ is a vector of size $|V|$ with each component as 1.
\end{enumerate} 
\end{enumerate}
\label{th:struc}
\end{theorem}

Theorem \ref{th:struc} considers the two objectives of the government for each attitude (pessimistic or optimistic) in the four parts of the theorem. The main finding is that the value for each of these objectives can be written as a product of a term that depends only on the network and a term that depends only on the utility function of the agent. This decomposition helps separate out the problems of information and incentive design from network design, helping us surpass the complexities arising due to a heterogeneous network and the associated externalities. Since informative signals do not change the underlying structure of the network, the signals only affect the component consisting of the unilateral effort, thus effectively reducing the complicated objectives of the government to simpler objectives which are only dependent on the unilateral effort $e^*$. Note that the network dependent terms in the structural characterization are quite different in each of the four cases. However, the terms related to the utility function are similar, giving two main reductions of the objectives: $e^*(\mu)$ and $b(e^*(\mu))$. Every objective either has one of these as the network independent term, or is a positive linear combination of both. Also note that, unlike the first three parts in the theorem, the fourth part contains results only in certain limits of $\sigma_b$. 

These results are obtained by making a key use of the results by Pandit and Kulkarni \cite{parthe}, which study the performance of the different equilibria in a public goods game. Their relevant results are stated below.

\begin{lemma} 
\begin{enumerate}
    \item \cite[Theorem 3(a)]{parthe} Among all NE, specialized equilibrium corresponding to the $w$-weighted maximum independent set attains the maximum $w$-weighted aggregate effort $\sum_{k\in V}w_k\cdot x_k^*$.
    \item \cite[Theorem 3(c)]{parthe} Whenever a distributed equilibrium exists, it leads to the minimum aggregate cost (defined as $c\cdot\sum_{k\in V} x^*_k$) among all \NE. Moreover, all distributed equilibria lead to the same aggregate cost.
    \item \cite[Equations 11, 12]{parthe} Bounds on the benefit received by an individual in equilibrium \begin{align*}
        b(e^*) &\leq b(\mathcal{E}_k(x^*)) \leq b(ne^*) \\
        b(e^*) + c\sigma_b\cdot(\mathcal{E}_k(x^*) - e^*) &\leq b(\mathcal{E}_k(x^*)) \leq b(e^*) + c\cdot(\mathcal{E}_k(x^*) - e^*)
    \end{align*} where $\sigma_b := \frac{b(ne^*)- b(e^*)}{ce^*(n-1)}$.
    \item \cite[Equations 13, 14]{parthe} Benefit received by an individual in equilibrium in limiting cases of $\sigma_b$  \begin{align*}
        \lim_{\sigma_b\to 0} b(\mathcal{E}_k(x^*)) &= b(e^*) \\
        \lim_{\sigma_b\to 1} b(\mathcal{E}_k(x^*)) &= b(e^*) + c\cdot(\mathcal{E}_k(x) - e^*).
    \end{align*} 
\end{enumerate} \label{lem:parthe}
\end{lemma}

While the detailed proof of how Theorem \ref{th:struc} follows from Lemma \ref{lem:parthe} can be found in the appendix, an intuitive understanding can be obtained by considering part 1 of the theorem. Pandit and Kulkarni show that the specialized equilibrium corresponding to the maximum independent set maximizes the aggregate effort among all Nash equilibria. The aggregate effort in this case is given by $\alpha(G)\cdot e^*$, where $\alpha(G)$ is size of the maximum independent set of the graph and $e^*$ is the unilateral effort \eqref{eq:net_eff}. Using this and the fact that for every belief distribution $\mu_\pi$ and for every signal $\iota_S$, a distinct public goods game is played at the level of the individuals, part 1 of Theorem \ref{th:struc} can be obtained.

\subsection{Optimal signaling policy}
We now present our results related to the optimal signaling policy. Let $\mathcal{A}(Q)$ be the coefficient of absolute risk aversion of a function $Q:[0,\infty)\rightarrow [0,\infty)$, defined as \cite{arrow1965aspects, pratt1978risk} $$\mathcal{A}(Q):= -\frac{Q''}{Q'}.$$
Let $\mathcal{P}(Q)$ be the coefficient of absolute prudence of a function $Q$ defined as \cite{eeckhoudt1995risk, maggi2006risk}
$$\mathcal{P}(Q) := -\frac{Q'''}{Q''}.$$
Let
\begin{align*}
    \mathcal{R}(b;\mu)&:=\mathcal{A}(\Delta b)-\frac{\mathcal{P}(b)}{2} \\
    \tilde{\mathcal{R}}(b;\mu) &:= \mathcal{A}(\Delta b)-\frac{\mathcal{P}(b)+\mathcal{A}(b)}{2},
\end{align*}

Recall the full disclosure and no disclosure policies as defined in \eqref{eq:fulldisc}, \eqref{eq:nodisc}.

\begin{theorem} (Full disclosure and no disclosure as the optimal signaling policy)

\begin{enumerate}
    \item (Optimistic government maximizing aggregate effort)
    
    Consider an optimistic government with objective given by \eqref{eq:obj1}. In the Stackelberg equilibrium (cf. Def \ref{def:OSE}),
    
     If $\mathcal{R}(b;\mu)$ is positive for all induced beliefs $\mu$, then no information disclosure is optimal. If $\mathcal{R}(b;\mu)$ is negative $\forall\mu$, then full information disclosure is optimal.
    \item (Pessimistic government maximizing aggregate effort)
    
    Consider a pessimistic government with objective given by \eqref{eq:obj1}. In the Stackelberg equilibrium (cf. Def \ref{def:PSE}), 
    
    If $\mathcal{R}(b;\mu)$ is positive for all induced beliefs $\mu$, then no information disclosure is optimal. If $\mathcal{R}(b;\mu)$ is negative $\forall\mu$, then full information disclosure is optimal. 
    \item (Pessimistic government maximizing probability of random individual being safe) 
    
    Consider a pessimistic government with objective given by \eqref{eq:obj2}. In the Stackelberg equilibrium (cf. Def \ref{def:PSE}),
    
     If $\tilde{\mathcal{R}}(b;\mu)$ is positive for all induced beliefs $\mu$, then no information disclosure is optimal. If $\tilde{\mathcal{R}}(b;\mu)$ is negative $\forall\mu$, then full information disclosure is optimal.
    \item (Optimistic government maximizing probability of random individual being safe) 
    
    Consider an optimistic government with objective given by \eqref{eq:obj2}. In the Stackelberg equilibrium (cf. Def \ref{def:OSE}), 
    
    \begin{enumerate}
        \item As $\sigma_b\to 0$, when $\tilde{\mathcal{R}}(b;\mu)>0$, no information disclosure is the optimal policy. When $\tilde{\mathcal{R}}(b;\mu)<0$, full information disclosure is optimal.
        \item As $\sigma_b\to 1$, when $\tilde{\mathcal{R}}(b;\mu)>0$, no information disclosure is the optimal policy. When $\mathcal{R}(b;\mu)<0$, full information disclosure is optimal.
    \end{enumerate}
\end{enumerate}
\label{th:policy}
\end{theorem}

Theorem \ref{th:policy} states that the optimal information disclosure policy is dependent on the sign of the quantities $\mathcal{R}(b;\mu)$ and $\tilde{\mathcal{R}}(b;\mu)$, which depend on the risk aversion of the differential probabilities and probabilities of being safe i.e., $\mathcal{A}(\Delta b)$ and $\mathcal{A}(b)$ and the prudence $\mathcal{P}(b)$. When these signs do not hold uniformly in one direction, a policy in between full disclosure and no disclosure may be optimal.

The results in Theorem \ref{th:policy} are obtained using our structural decomposition results in Section \ref{subsec:res_struc} (which show that the government objectives reduce to $e^*$ and $b(e^*)$) and the concavification arguments discussed by Kamenica and Gentzkow \cite{persuasion} (which show that when government objectives are strictly concave (convex) with respect to public beliefs $\mu$, no (full) disclosure is the optimal policy). The detailed proofs can be found in the appendix \ref{sec:Analysis}.

Furthermore, we can see that when full information is the optimal signaling policy for maximizing $e^*$, it is also the optimal policy for maximizing the other two objectives ($b(e^*)$ and positive linear combination of $e^*$ and $b(e^*)$). When no information disclosure is the optimal policy for maximizing $b(e^*)$, it is also the optimal policy for maximizing the other two objectives. This follows from the risk aversion conditions for optimizing $e^*$ and $b(e^*)$ (Theorem \ref{th:policy}) and the fact that $\mathcal{A}(b)>0$ (since $b$ is increasing and concave).

Finally, as a corollary to Theorem \ref{th:policy}, we obtain sufficient conditions, which are independent of the belief of individuals $\mu$ and only dependent on the functions $b(x;h)$ and $b(x;l)$, under which full disclosure and no disclosure is optimal for the aggregate effort objective of the government (for both the attitudes). 

\begin{corollary}
(Sufficient conditions for optimal information disclosure policy, independent of $\mu$, for the aggregate effort objective of the government) 

Let $
     Y(b(\cdot;h), b(\cdot;l)):=2\frac{(\Delta b)''}{(\Delta b)'} (\Delta b)'' - (\Delta b)'''$ and $ Z(b(\cdot;h), b(\cdot;l)):=2\frac{(\Delta b)''}{(\Delta b)'} b''(x;l) - b'''(x;l)$.
When $Y(b(\cdot;h), b(\cdot;l))<0$ and $Z(b(\cdot;h), b(\cdot;l))>0$, no information is the optimal policy. On the other hand, when $Y(b(\cdot;h), b(\cdot;l))>0$ and $Z(b(\cdot;h), b(\cdot;l))<0$, full information is the optimal policy. \label{cor:suff_cond}
\end{corollary}

The sufficient conditions in Corollary \ref{cor:suff_cond} are useful since they allow us to talk about the optimal information disclosure by looking at the nature of the benefit functions alone, without requiring details about individuals' beliefs.

Next, we look at cases where two special intermediate policies, exaggeration and downplay, are the optimal disclosure policies. The exaggeration signaling policy is one in which the high infection state is always signaled correctly, but the low infection level is also signaled as high with some non-negative probability i.e.,
\begin{align}
    \text{Exaggeration:} \hspace{2em} \pi(\iota_S=h|\iota=h)=1, \hspace{1em} \pi(\iota_S=h|\iota=l)>0. \label{eq:exag}
\end{align}
On the other hand, the downplaying policy is one in which the low infection state is always correctly signaled, but the high infection state is signaled as low with some non-negative probability i.e.,
\begin{align}
    \text{Downplay:} \hspace{2em} \pi(\iota_S=l|\iota=l)=1, \hspace{1em} \pi(\iota_S=l|\iota=h)>0. \label{eq:down}
\end{align}

\begin{theorem} (Exaggeration and downplay as the optimal signaling policy)

\begin{enumerate}
    \item (Optimistic government maximizing aggregate effort)
    
    Consider an optimistic government with objective given by \eqref{eq:obj1}. In the Stackelberg equilibrium (cf. Def \ref{def:OSE}),
     
     If $\exists $ $\mu_e$ such that $\forall \mu< \mu_e, \mathcal{R}(b;\mu)<0$ and $\forall \mu > \mu_e, \mathcal{R}(b;\mu)>0$, then exaggeration is the optimal policy. On the other hand, if $\exists$ $ \mu_d$ such that $\forall \mu< \mu_d, \mathcal{R}(b;\mu)>0$ and $\forall \mu > \mu_d, \mathcal{R}(b;\mu)<0$, then downplaying the infection level is optimal.
    \item (Pessimistic government maximizing aggregate effort)
    
    Consider an pessimistic government with objective given by \eqref{eq:obj1}. In the Stackelberg equilibrium (cf. Def \ref{def:PSE}),
    
    If $\exists $ $ \mu_e$ such that $\forall \mu< \mu_e, \mathcal{R}(b;\mu)<0$ and $\forall \mu > \mu_e, \mathcal{R}(b;\mu)>0$, then exaggeration is the optimal policy. On the other hand, if $\exists $ $ \mu_d$ such that $\forall \mu< \mu_d, \mathcal{R}(b;\mu)>0$ and $\forall \mu > \mu_d, \mathcal{R}(b;\mu)<0$, then downplaying the infection level is optimal.
    \item (Pessimistic government maximizing probability of random individual being safe) 
    
    Consider an pessimistic government with objective given by \eqref{eq:obj2}. In the Stackelberg equilibrium (cf. Def \ref{def:PSE}),
    
     If $\exists$ $ \mu_e$ such that $\forall \mu< \mu_e, \tilde{\mathcal{R}}(b;\mu)<0$ and $\forall \mu > \mu_e, \tilde{\mathcal{R}}(b;\mu)>0$, then exaggeration is the optimal policy. On the other hand, if $\exists $ $\mu_d$ such that $\forall \mu< \mu_d, \tilde{\mathcal{R}}(b;\mu)>0$ and $\forall \mu > \mu_d, \tilde{\mathcal{R}}(b;\mu)<0$, then downplaying the infection level is optimal.
    \item (Optimistic government maximizing probability of random individual being safe) 
    
    Consider an optimistic government with objective given by \eqref{eq:obj2}. In the Stackelberg equilibrium (cf. Def \ref{def:OSE}),
    \begin{enumerate}
        \item As $\sigma_b\to 0$, if $\exists $ $\mu_e$ such that $\forall \mu< \mu_e, \tilde{\mathcal{R}}(b;\mu)<0$ and $\forall \mu > \mu_e, \tilde{\mathcal{R}}(b;\mu)>0$, then exaggeration is the optimal policy. On the other hand, if $\exists $ $ \mu_d$ such that $\forall \mu< \mu_d, \tilde{\mathcal{R}}(b;\mu)>0$ and $\forall \mu > \mu_d, \tilde{\mathcal{R}}(b;\mu)<0$, then downplaying the infection level is optimal.
        \item As $\sigma_b\to 1$, if $\exists $ $\mu_e$ such that $\forall \mu< \mu_e, \mathcal{R}(b;\mu)<0$ and $\forall \mu > \mu_e, \tilde{\mathcal{R}}(b;\mu)>0$, then exaggeration is the optimal policy. On the other hand, if $\exists $ $\mu_d$ such that $\forall \mu< \mu_d, \tilde{\mathcal{R}}(b;\mu)>0$ and $\forall \mu > \mu_d, \mathcal{R}(b;\mu)<0$, then downplaying the infection level is optimal.
    \end{enumerate}
\end{enumerate}
\label{th:partial policy}
\end{theorem}

Theorem \ref{th:partial policy} shows that a single flip of the signs of $\mathcal{R}(b;\mu)$ and $\tilde{\mathcal{R}}(b;\mu)$ as a function of $\mu$ can lead to the intermediate policies of exaggeration and downplay being optimal. Whether the flip happens from positive to negative, or negative to positive determines which of these two policies would be optimal. The proof of the theorem follows from our structural characterization results in Theorem \ref{th:struc} and the concavification arguments of Kamenica and Gentzow \cite{persuasion}. Using their concavification arguments, we can argue that when the government objectives are first strictly convex (concave) and then strictly concave (convex) in the public's belief $\mu$, exaggeration (downplay) is the optimal policy. The detailed proof of Theorem \ref{th:partial policy} can be found in the appendix \ref{sec:Analysis}.

Next, we show that there is an inherent symmetry in the optimal policies of the government for a two-state system like ours.

\begin{prop}(Information invariance and symmetry in policies)

Let $p_h:=\pi(\iota_S=h|\iota=h)$ and $p_l:=\pi(\iota_S=l|\iota =l)$. $(p_l,p_h)\in [0,1]\times[0,1]$ completely describes the policy $\pi$. Consider two policies $(a,b)$ and $(1-a,1-b)$. Then:
\begin{enumerate}
    \item For policy $(a,b)$ and signal $\iota_S$, the same distribution of posterior beliefs are generated by policy $(1-a,1-b)$ and signal $\iota_S^c$.
    \item Policy $(a,b)$ and signal $\iota_S$ generates the same best response among individuals as policy $(1-a,1-b)$ and signal $\iota_S^c$. 
    \item The expected objective, given the real infection level remains the same under both the policies $(a,b)$ and $(1-a,1-b)$.
\end{enumerate}\label{prop:invariance}
\end{prop}

Having such information invariance would mean that we can restrict the domain space of policies to only the half of the current space i.e., the domain defined by $p_h, p_l \geq0$ and $p_h+p_l \leq 1$ is sufficient. Figure \ref{fig:diff_policies} gives a pictorial representation of the full disclosure, no disclosure, exaggeration and downplay policies on the $(p_l,p_h)$ plot and summarises our results in Theorems \ref{th:policy} and \ref{th:partial policy}. The invariance discussed in Proposition \ref{prop:invariance} is also depicted in the figure.

\pgfplotsset{
  /pgfplots/xlabel near ticks/.style={
     /pgfplots/every axis x label/.style={
        at={(ticklabel cs:0.5)},anchor=near ticklabel
     }
  },
  /pgfplots/ylabel near ticks/.style={
     /pgfplots/every axis y label/.style={
        at={(ticklabel cs:0.5)},rotate=90,anchor=near ticklabel}
     }
  }

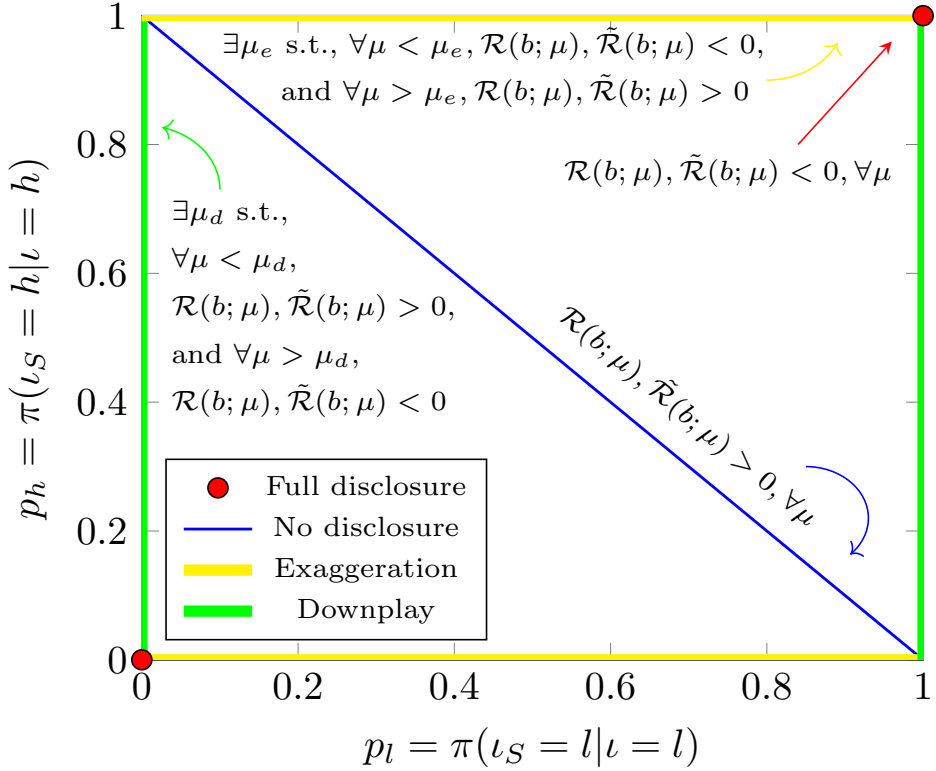
\begin{figure}[htp]
\begin{center}
    \begin{tikzpicture}[scale=1.5]
\begin{axis}[
xmin = 0, xmax = 1,
ymin = 0, ymax = 1,
legend pos=south west,
restrict y to domain=0:1,
restrict x to domain=0:1,
    ylabel near ticks,
    xlabel near ticks,
    xlabel={$p_l = \pi(\iota_S=l|\iota=l)$},
    ylabel={$p_h = \pi(\iota_S=h|\iota=h)$},
]

\addplot[mark=*,only marks, fill=red, mark size=2.5pt] coordinates {(0,0)} node[below, pos=1]{};
\node[align=left] at (0.75,0.76) {\scriptsize $\mathcal{R}(b;\mu), \tilde{\mathcal{R}}(b;\mu)<0, \forall \mu$}; 
\node[align=left, rotate=-40] at (0.7,0.38) {\scriptsize $\mathcal{R}(b;\mu), \tilde{\mathcal{R}}(b;\mu)>0, \forall \mu$}; 
\node[align=left] at (0.45,0.92) {\scriptsize $\exists \mu_e$ s.t., $\forall \mu< \mu_e, \mathcal{R}(b;\mu), \tilde{\mathcal{R}}(b;\mu)<0$, \\ \scriptsize \hspace{1.5em} and $\forall \mu>\mu_e, \mathcal{R}(b;\mu), \tilde{\mathcal{R}}(b;\mu)>0$ }; 
\node[align=left] at (0.22,0.545) {\scriptsize $\exists \mu_d$ s.t., \\ \scriptsize $\forall \mu< \mu_d$, \\ \scriptsize $\mathcal{R}(b;\mu), \tilde{\mathcal{R}}(b;\mu)>0$, \\ \scriptsize and $\forall \mu>\mu_d,$ \\ \scriptsize $\mathcal{R}(b;\mu), \tilde{\mathcal{R}}(b;\mu)<0$ }; 
\addplot[color=blue,domain=0:1,samples=1000,smooth, line width=0.25mm] {-x+1} node[right,pos=1] {};
\addplot[color=yellow,domain=0:1,samples=1000,smooth, line width = 1mm] {0} node[right,pos=1] {};
\addplot[line width=1mm, samples=50, smooth,domain=0:6,green] coordinates {(1,0)(1, 1)};
\addplot[line width=1mm, samples=50, smooth,domain=0:6,green] coordinates {(0,0)(0, 1)};
\addplot[color=yellow,domain=0:1,samples=1000,smooth, line width=1mm] {1} node[right,pos=1]{};
\addplot[mark=*,only marks, fill=red, mark size=2.5pt] coordinates {(1,1)} node[below, pos=1]{};
\draw [-stealth, red](0.84,0.8) -- (0.96,0.96);
\draw[-> , green] (0.1,0.73) arc (0:75:0.1);
\draw[-> , yellow] (0.8,0.9) arc (270:335:0.1);
\draw[-> , blue] (0.85,0.3) arc (90:-45:0.08);
   \addlegendentry{\scriptsize Full disclosure}
   \addlegendentry{\scriptsize No disclosure}
   \addlegendentry{\scriptsize Exaggeration}
   \addlegendentry{\scriptsize Downplay}
\end{axis}
\end{tikzpicture}
\end{center}
\caption{Summary of results in Theorem \ref{th:policy} and \ref{th:partial policy}. Every point $(p_l,p_h)$ in this plot corresponds to a possible signaling policy of the government. The class of no disclosure policies correspond to the diagonal corresponding to $p_h+p_l=1$ (marked in blue), while full disclosure corresponds to the points $(p_l,p_h)=(1,1),(0,0)$ (marked in red). Exaggeration (yellow) and downplay (green) policies correspond to the lines $p_h=1$ and $p_l=1$ respectively, and by Proposition \ref{prop:invariance}, also to lines $p_h=0$ and $p_l=0$. The conditions under which each of the policies are optimal for different government objectives are dependent on the sign of the functions $\mathcal{R}$ and $\tilde{\mathcal{R}}$, and are written beside each policy in the figure.}
     \label{fig:diff_policies}
\end{figure}

\subsection{Illustration of main results}

We now illustrate these results for some examples of benefit functions. Before we begin, recall the desired conditions on these functions. The benefit functions $b(x; h)$ and $b(x; l)$ are increasing and concave in the argument $x$. Moreover, since these functions give us the probability of being safe in the high and low infection states, we require them to be non-negative and to saturate to $1$ as the effort $x$ goes to $\infty$. Also, since one is safer in the low infection state than the high infection state, we have $b(x; h)< b(x; l)$, $\forall x$. 

The parameters in the examples below are chosen appropriately to ensure all these conditions are satisfied and that Assumption \ref{asu:1} holds. Example \ref{ex:1} also satisfies Assumption \ref{asu:2}, whereas Example \ref{ex:2} shows a case in which Assumption \ref{asu:2} is relaxed.

\begin{example} \label{ex:1}
Consider exponential benefit functions given by \begin{align*}
    b(x; h) = 1-He^{-x}, \hspace{2em}
    b(x; l) =  1-Le^{-x},
\end{align*} with $H>L$ and $c <L$ (the latter ensures that Assumption \ref{asu:2} holds). 

    Since the risk aversion and prudence of exponential functions are known to be constants, we can easily see that $\mathcal{R}(b;\mu)>0$ and $\tilde{\mathcal{R}}(b;\mu)=0$ for the above benefit function. This means that $e^*(\mu)$ is strictly concave, leading to no disclosure being the optimal policy for governments maximizing aggregate effort (objective \eqref{eq:obj1}) with any attitude. This is because, by Theorem \ref{th:struc}, objective \eqref{eq:obj1} is proportional to the expected value of $e^*$. On the other hand, $b(e^*(\mu))$ is independent of $\mu$. This causes all policies to be equally effective for governments with objective \eqref{eq:obj2} which reduce to expected $b(e^*)$ through the structural characterization (cf. Theorem \ref{th:struc}, points 3 and 4(a)). Since the final objective in Theorem \ref{th:struc} part 4(b) is a linear combination of $e^*$ and $b(e^*)$, no disclosure is a common signaling policy which is optimal for all the four government objective and attitude cases considered.
    This can also be seen using the sufficient conditions derived in Corollary \ref{cor:suff_cond}.

    \begin{figure}[htp]
     \centering
     \begin{subfigure}[b]{0.495\textwidth}
         \centering
         \includegraphics[width=\textwidth]{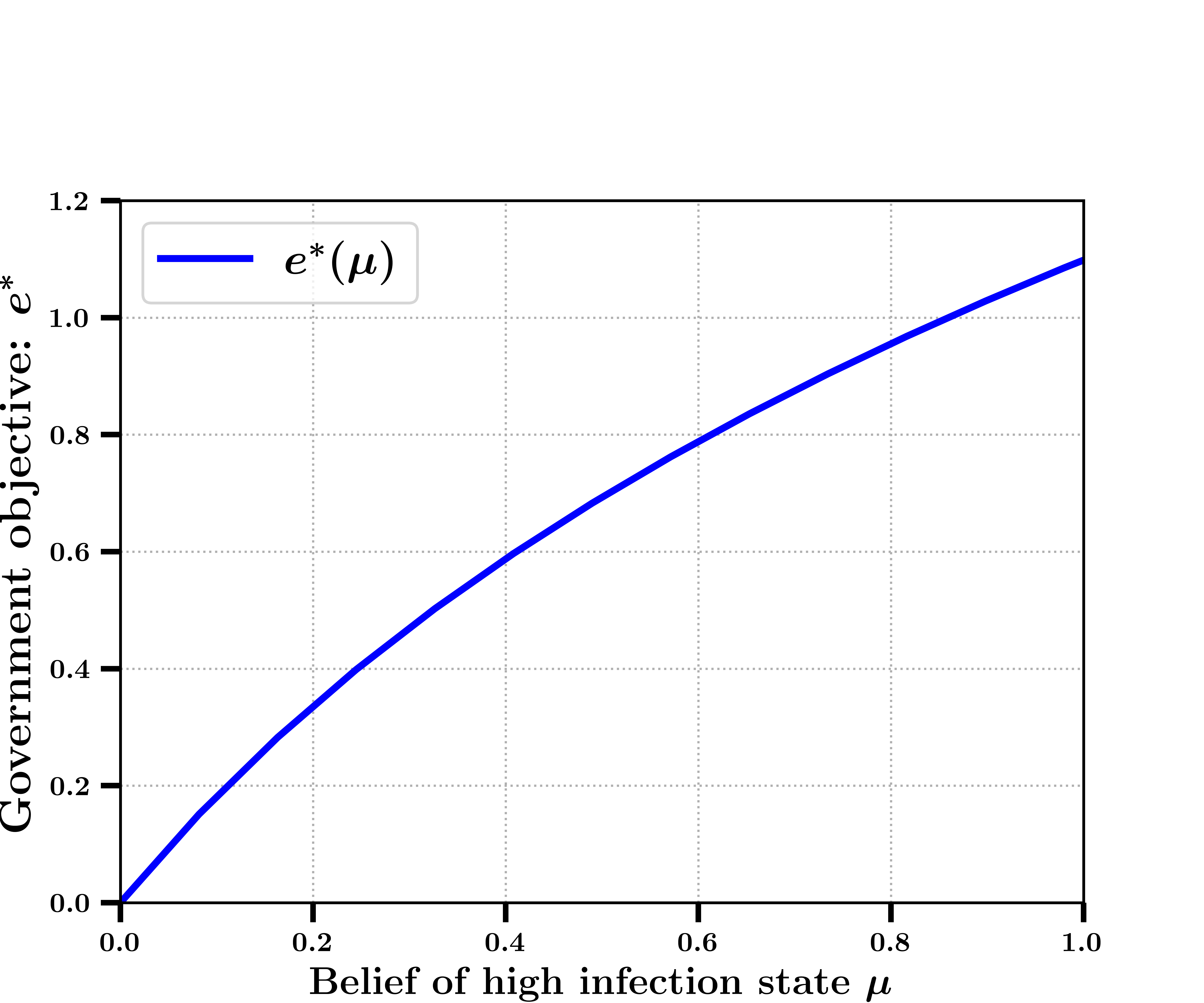}
     \end{subfigure}
     \hfill
     \begin{subfigure}[b]{0.495\textwidth}
         \centering
         \includegraphics[width=\textwidth]{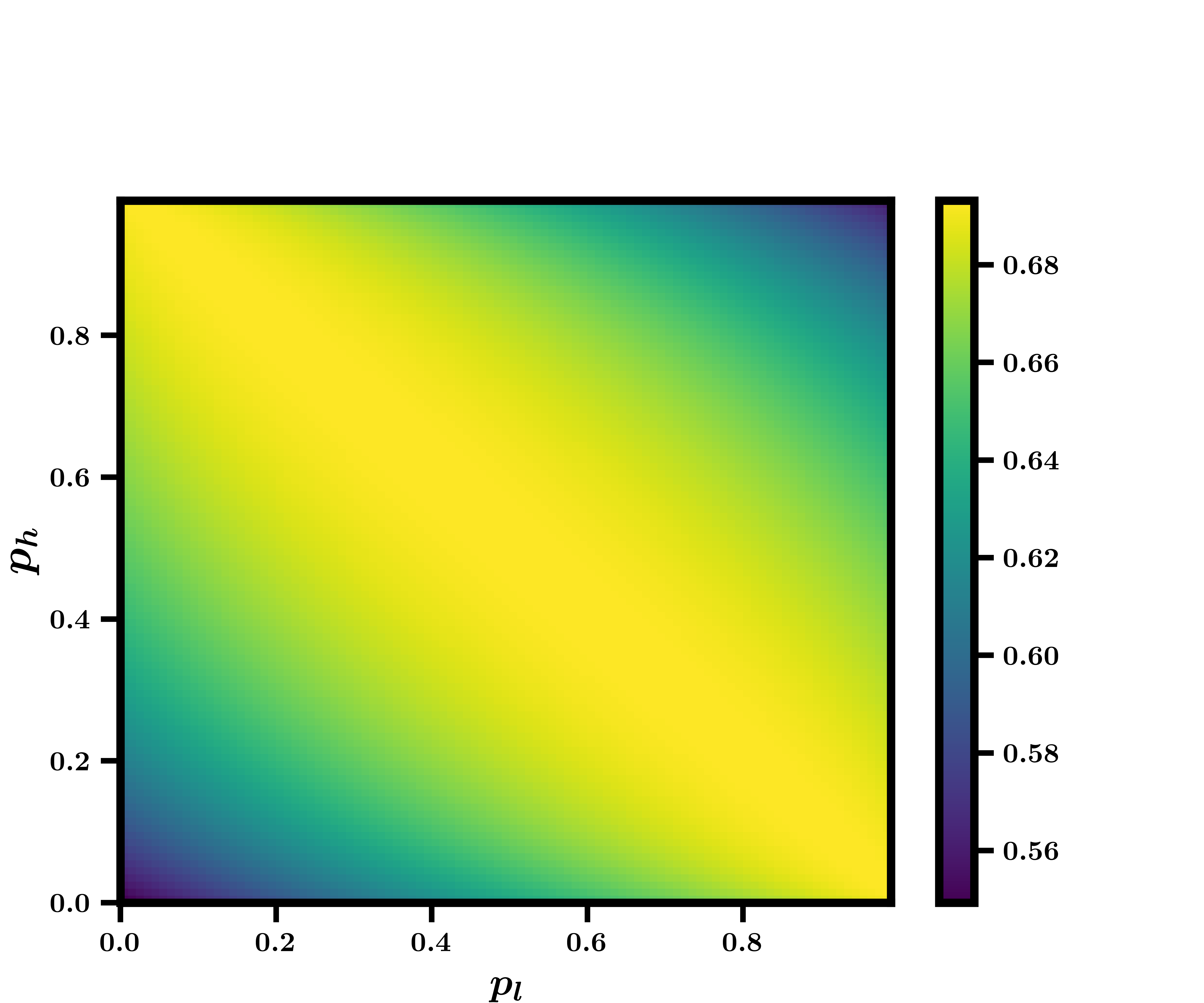}
     \end{subfigure}
        \caption{Plots for the benefit function considered in Example 1: (L) Curvature of the unilateral effort $e^*$ as a function of the belief of the high infection state $\mu$, which determines the optimal information disclosure policy. (R) Density plot for the expected unilateral effort $\mathbb{E}[e^*]$ as a function of the policy $\pi=(p_l,p_h)$.}
        \label{fig:ex1}
\end{figure}
    
    Figure \ref{fig:ex1}(L) shows that the objectives which reduce to expected $e^*$ as a function of $\mu$ are strictly concave for the benefit function considered in Example \ref{ex:1}. This means that no disclosure is the optimal policy. This can also be seen in Figure \ref{fig:ex1}(R), which is the density plot of the government's objective of maximizing $e^*$ on the $(p_l,p_h)$ policy plot as in Figure $\ref{fig:diff_policies}$. Notice that the maximum (yellow) is achieved on the line $p_h+p_l=1$, which is the line of `no disclosure'.
    
\end{example}

\begin{example} \label{ex:2}
Consider the same benefit functions as before \begin{align*}
    b(x; h) = 1-He^{-x}, \hspace{2em}
    b(x; l) =  1-Le^{-x}
\end{align*} with $H>L$ but $c>L$. Notice that for such $c$ Assumption \ref{asu:2} does not hold, whereby Theorems \ref{th:policy} and \ref{th:partial policy} do not directly apply. Since the boundary conditions on the effort $e^*$ dictate it to be non-negative, $e^*(\mu)$ is no longer strictly concave (see Figure \ref{fig:ex2}(L)). From the density plot in Figure \ref{ex:2}(R), we see that the maximum (yellow) occurs on a point on the line $p_h=1$ and $p_h=0$, which shows that the optimal policy lies in the class of exaggeration policies. This can also be analytically calculated by applying the techniques used for proving Theorem \ref{th:partial policy}. As in Example \ref{ex:1}, $b(e^*(\mu))$ is independent of $\mu$, and all policies perform equally well. The exaggeration policy obtained above is therefore a common signaling policy optimal for all the four government objective and attitude cases considered.

\begin{figure}[htp]
     \centering
     \begin{subfigure}[b]{0.495\textwidth}
         \centering
         \includegraphics[width=\textwidth]{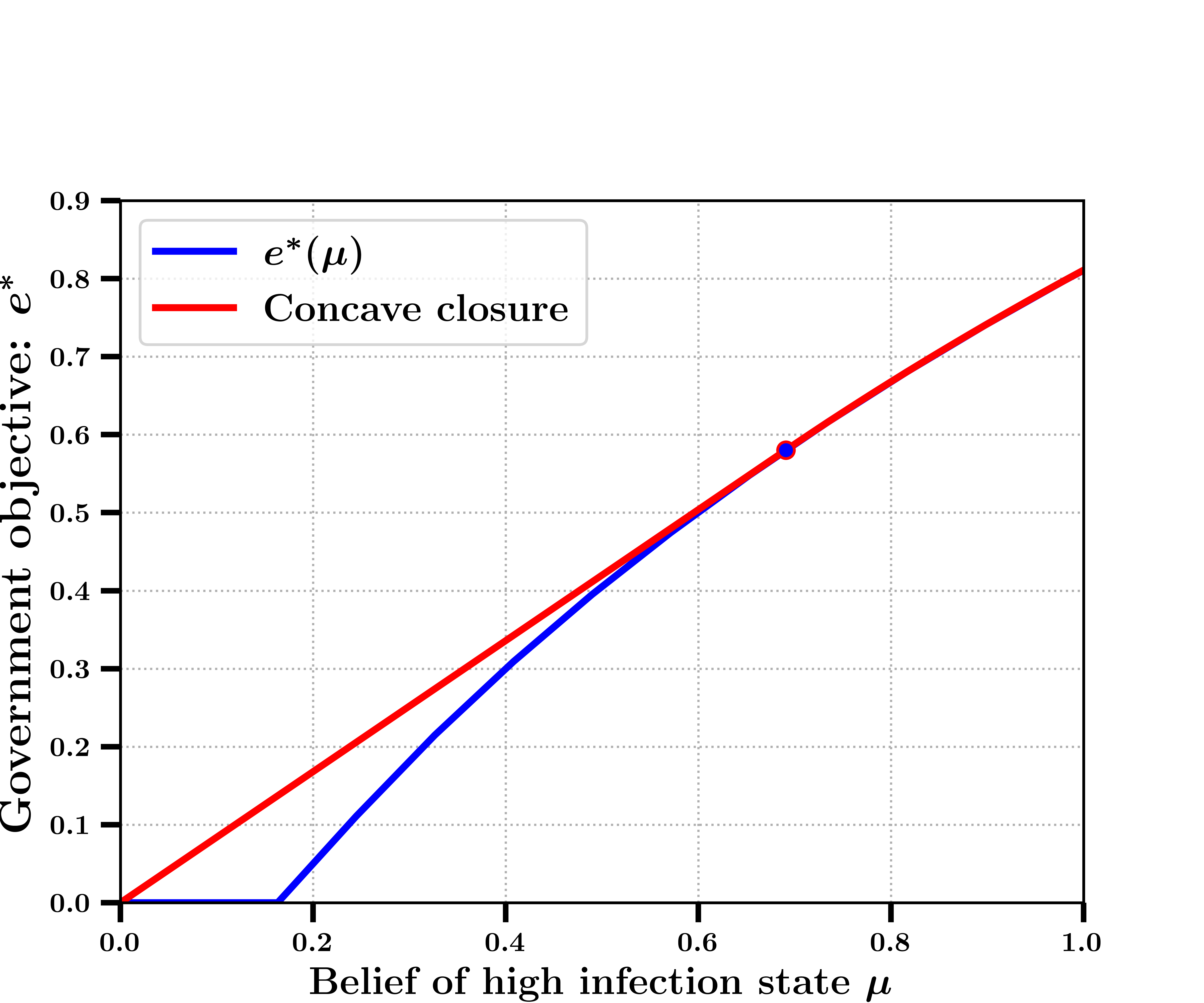}
     \end{subfigure}
     \hfill
     \begin{subfigure}[b]{0.495\textwidth}
         \centering
         \includegraphics[width=\textwidth]{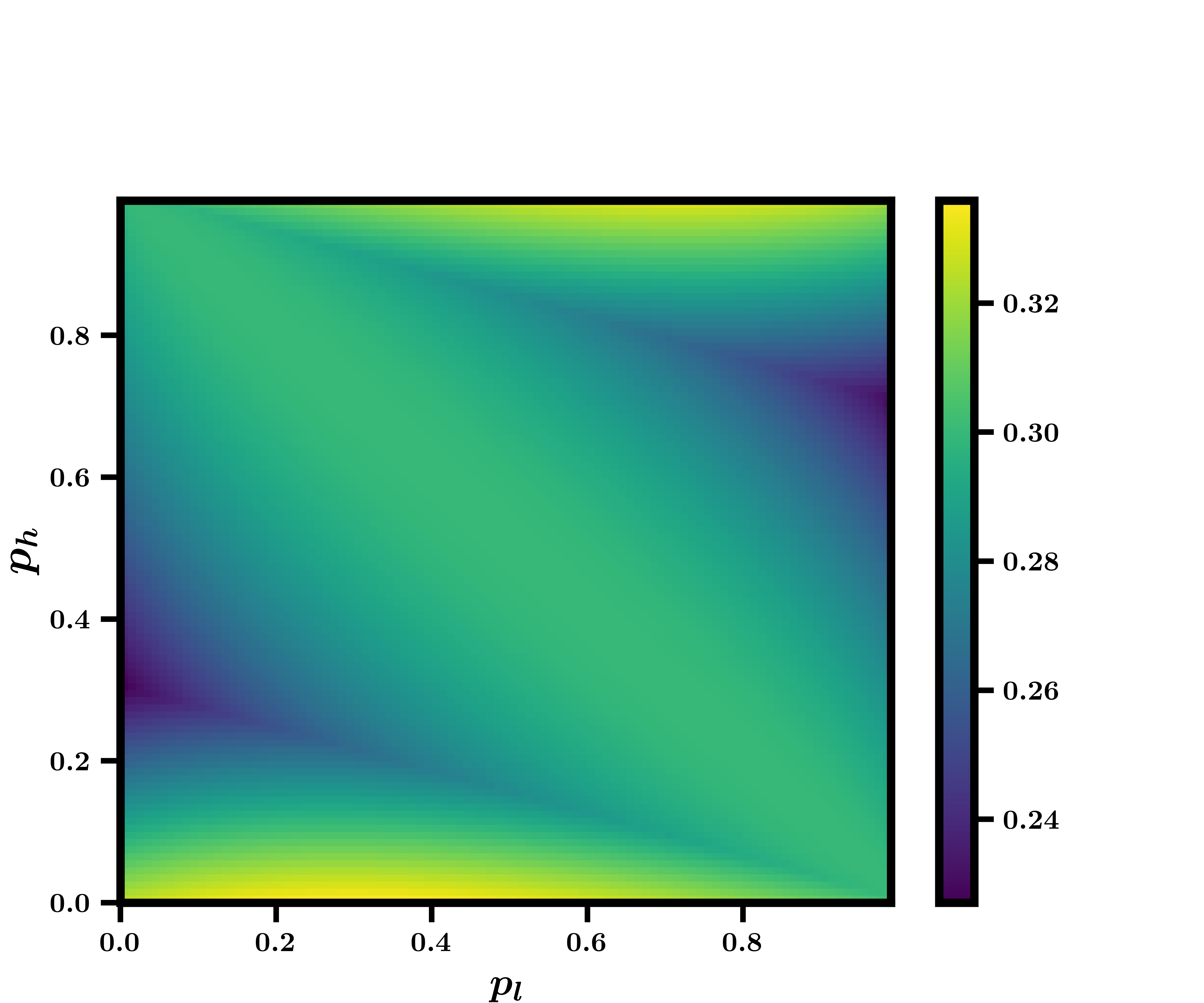}
     \end{subfigure}
        \caption{Plots for the benefit functions in Example 2: (L) Curvature of the unilateral effort $e^*$ as a function of the belief of the high infection state $\mu$; (R) Density plot for the expected unilateral effort $\mathbb{E}[e^*]$ as a function of the policy $\pi=(p_h,p_l)$.}
        \label{fig:ex2}
\end{figure}
\end{example}

\section{Discussion and Conclusion} \label{sec:discussion}



\subsection{Policy implications of equilibrium analysis and structural characterization of objectives}
Recall the performance of the different equilibria of the epidemic game for both the government objectives from Theorem \ref{th:struc}. The equilibrium performing the best among all Nash equilibria is a specialized equilibrium (Theorem \ref{th:struc}, parts 1 and 4). Whereas, whenever a distributed equilibrium exists, it performs the worst among all Nash equilibria for both the objectives (Theorem \ref{th:struc}, parts 2 and 3). Therefore, government policies which force every individual to take at least some preventive measure, either forces the system out of equilibrium (since distributed equilibria might not always exist) or takes the system to a socially sub-optimal state for the society.

The structural characterization of the objectives allows us to independently look at two other kinds of intervention design, in addition to information design: incentive design (direct incentives in the form of transfers, subsidies, fines which change actions of individuals through change of preferences by tempering with their marginal utilities) and network design (constraining the connections among individuals). Incentives also affect the component in the structural characterization which is independent of the network. On the other hand, network design interventions affect the graph-theoretic terms in the structural characterizations.

In the epidemic management setting, incentives affect the cost of the effort taken, which is captured in our model through the marginal cost of effort parameter $c$. The unilateral effort $e^*$ changes with $c$ according to \eqref{eq:net_eff}, i.e.,
\begin{align*}
    b'(e^*) = c.
\end{align*} Since $b'$ is a decreasing function, the effort $e^*$ is inversely related to the marginal cost of effort $c$. Thus, transfers which reduce $c$ increase $e^*$ (and also $b(e^*)$ since $b$ is an increasing function) and lead to better societal outcomes, as one would anticipate.

Lockdowns and social distancing norms announced by governments are examples of network design interventions. Unlike incentive design and information design, network design affects the network dependent term in the structural characterization of the objectives. Note that interventions which arbitrarily cut connections among individuals may not have an impact on the welfare of the society, and only those network design interventions which affect the network dependent term in the structural characterization are useful. For example, in the optimistic case for maximizing the aggregate effort, the network dependent term is the maximum independent set (Theorem \ref{th:struc} part 1). Thus, a graph with a larger maximum independent set would perform better than one with a smaller maximum independent set. Among all connected graphs on $n$ nodes, a star network has the largest maximum independent set and thus performs the best for this objective. The other objectives have a different dependence on the  network properties, as seen in Theorem \ref{th:struc}, leading to different choices of optimal networks for each of them. 
The study of the structure of these optimal networks could potentially also contribute to other epidemic related aspects such as decisions regarding distribution of essential services to public during epidemics. 


\subsection{Optimal signaling policies and robustness of results}
The results related to the optimal information disclosure policy are linked to the risk aversion and prudence of the probability of being safe, as a function of their effort. The results imply that when the government objective is proportional to expected $e^*$, only the prudence of individuals play a role, while if the government objective is proportional to expected $b(e^*)$, both risk aversion and prudence play an equal role in determining the optimal policy. 

From the structural characterization of the objectives (Theorem \ref{th:struc}), we see that the government objectives effectively reduce to either expected $e^*$, expected $b(e^*)$, or positive linear combinations of these. Moreover, we see that there exists a dependence between the optimal disclosure policies for the objectives which reduce to $e^*$ with those that reduce to $b(e^*)$ -- when full disclosure was optimal for $e^*$, it is also optimal for the other objectives and when no disclosure was optimal for $b(e^*)$, it is also optimal for the other objectives. This gives us at least certain cases under which different government objectives share a common optimal policy. We also see this in the illustrative examples. This is an important property that an ideal policy should have since governments often have multiple qualitatively similar, but quantitatively different objectives. It ensures that the nature of the optimal policy is robust to these different government objectives.

Moreover, since the network dependent terms form a separate component in the structural characterization which is unaffected by the public signaling policies, our results are robust to different network structures and are thus also scalable to very large networks.

\subsection{Conclusions}
In this paper, we considered a problem of information design to change the perception of interacting individuals in the society, thereby nudging them to behave pro-socially during epidemic. We formulated a model that captures this real-world process at two levels: at the level of the individual (formulated as a networked public goods game \cite{bramoulle}) and at the level of the government (modelled using the framework of Bayesian persuasion \cite{persuasion}). 
	
We studied the problem of intervention design for two possible objectives of the government (maximizing the society's total actions taken to stay safe, and maximizing the probability of a random individual being safe during the epidemic) for two government attitudes (optimistic and pessimistic). Despite the complexity of the setting, we discovered a structural result that helped in the design of information disclosure policies. The structural result showed a decomposition in the government's objective, leading to them having two independent components: a network dependent term and another depending on the utility function alone. Using this characterization, we found conditions under which different policies are optimal. We saw that these conditions depend on the risk aversion and prudence of the benefit functions. We also found sufficient conditions, which depend on the benefit functions in the high and low infection state alone with no dependence on the public's beliefs, for which the optimal policies were to disclose no information and to disclose full information to the public.

All results in this paper focus on situations where the same signal is sent to every individual. For a long time in history, these were the only types of signals that could be sent since the communication between the government and the society was limited to broadcasting media like television and radio. With the advent of mobile applications in the recent past, information dissemination can happen through personalized signals. These signals can either share information about the local infection level, or the signals can be such that they are privately observed. While some recent studies \cite{sitabhra, arieli2019private} explore these new avenues, a systematic analysis is still an open problem and an interesting future direction.

\bibliography{ref}

\appendix

\section{Analysis and proofs of results} \label{sec:Analysis}
\subsection{Preliminary analysis}
Recall the observation we made during our model formulation in equation \eqref{eq:epidemic game} where we establish that the expected utility equation of epidemic game at the level of the individuals is the same as the utility equation of the networked public goods game. Our results in Section \ref{sec:results} follow from the work by Pandit and Kulkarni \cite{parthe}, summarized in Lemma \ref{lem:parthe}. They derive these results by characterizing the equilibria of the public goods game introduced in \cite{bramoulle} as solution to the following Linear Complementarity Problem \cite{cottle1992linear}
\begin{align}
    \text{LCP:}\qquad x^*\geq 0,\quad (A_G+I_n)x^* \geq v^*, \quad
    (x^*)^\top ((A_G+I_n)x^*-v^*) = 0, \label{eq:LCP}
\end{align}
    where $A_G$ is the adjacency matrix of graph $G$, $I_n$ is an identity matrix of size $n$ and $v^*$ is a vector of size $n$ with each component equal to $e^*$.

Also recall the mathematical expression of government's objective of maximizing the probability of a randomly selected individual being safe in \eqref{eq:obj2}. Through simple arithmetic, we can see that
\begin{align*}
    \mathbb{E}_{\iota, \iota_S} b((\mathcal{E}_k(x))_{\iota_S}; \iota)=\mathbb{E}_{\iota,\iota_S} \tilde{b}((\mathcal{E}_k(x)); \mu_{\pi},\iota_S)
\end{align*} 
where $\tilde{b}((\mathcal{E}_k(x));\mu_\pi,\iota_S) := \mu_\pi(\iota=h|\iota_S)b_h((\mathcal{E}_k(x))_{\iota_S})+\mu_\pi(\iota=l|\iota_S)b_l((\mathcal{E}_k(x))_{\iota_S})$ is the benefit function in the resulting public goods game when signal $\iota_S$ is received, as defined in \eqref{eq:epidemic game}.


Thus, objective \eqref{eq:obj2} is equivalent to \begin{align}
    \max_\pi \mathbb{E}_{\iota, \iota_S}[\frac{1}{n} \sum_{k\in V}b((\mathcal{E}_k(x))_{\iota_S};\iota)] & = \frac{1}{n}\max_\pi\mathbb{E}_{\iota, \iota_S} [B(x_{\iota_S})] \label{eq:eqobj2}
\end{align} where  $B(x_{\iota_S}) := \sum_{k\in V} \tilde{b}(\mathcal{E}_k(x));\mu_\pi,\iota_S)$ is the aggregate benefit in this public goods game.
Hence the objective of maximizing the probability of a randomly selected individual being safe is the same as that of maximizing expected aggregate benefit in the resulting public goods game.

The following lemma studies the performance of distributed equilibria in public goods games for the objective of aggregate benefit.

\begin{lemma}
(Aggregate benefit performance of distributed equilibria in public goods game)

For a public goods game with utility equation given by $U_k(x) = b(\mathcal{E}_k(x)) - c\cdot x_k$, whenever a distributed equilibrium exists, it achieves the least possible aggregate benefit $B(x):=\sum_{k\in V}b(x_k+\sum_{j\in N_k}x_j)$ among all Nash equilibria. Moreover, for distributed equilibrium $B(x) = nb(e^*)$. \label{lem:2}
\end{lemma}
\begin{proof}
\begin{enumerate}
    \item From LCP \eqref{eq:LCP}, the collective neighbourhood effort $\mathcal{E}_k(x)\geq e^*$, $\forall k$ in equilibrium. Since $b$ is an increasing function, $\sum_k b(\mathcal{E}_k(x)) \geq nb(e^*)$.
    \item By definition, in a distributed equilibrium $x_k > 0$, $\forall k$. Therefore, from LCP \eqref{eq:LCP}, in a distributed equilibrium, $\mathcal{E}_k(x) = e^*$, $\forall k$. Thus, $\sum_k b(\mathcal{E}_k(x)) = nb(e^*)$.
\end{enumerate}
\end{proof}

Next, for studying the best Nash equilibria for the maximizing the aggregate benefit in a public goods game, we first find the following bounds. The proofs of Lemma \ref{lem:3} and Corollary \ref{cor:1} are adapted from the proofs of \cite[Theorem 2]{parthe}.

\begin{lemma}
(Bounds on the optimistic equilibrium for maximizing aggregate benefit in public goods game) 
\begin{align}
        nb(e^*)&\leq \max_{x\in \NE} B(x) \leq nb(ne^*) \label{eq:8}\\
        nb(e^*) + (\alpha_{(A+I)e} - n)\sigma_bce^*&\leq \max_{x\in \NE} B(x) \leq nb(e^*)  +(\alpha_{(A+I)e}-n)ce^* \label{eq:9}
    \end{align}
    where $\alpha_{(A+I)e}$ is the $(A_G+I_n)e$-weighted independence number. 

\label{lem:3}
\end{lemma}

\begin{proof}
This follows from taking summation over all nodes in inequalities of Lemma \ref{lem:parthe} point 3 and using the result of in Lemma \ref{lem:parthe} point 1 which states that maximum weighted aggregate effort is achieved by weighted specialized equilibria. The latter result is required in the following context: when the summation over all nodes is taken, the term $\sum_k (x_k + \sum_{j\in N_k} x_j)$ corresponds to the $(A_G+I_n)e$-weighted aggregate effort. The maximum of this term among all Nash equilibria $\max_{x\in \NE} \sum_k (x_k +\sum_{j\in N_k} x_j)$ is achieved by the specialized equilibrium of corresponding to the $(A_G+I_n)e$-weighted maximum independent set and is thus equal to $\alpha_{(A+I)e}e^*$.
\end{proof}

Using the bounds in Lemma \ref{lem:3}, in the following corollary we find the best performing Nash equilibria in the limits $\sigma_b\to 0$ and $\sigma_b \to 1$. 

\begin{corollary} (Limit of best Nash equilibria for maximizing aggregate benefit in public goods game)

As $\lim_{\sigma_b\to 0}$, the lower inequality in \eqref{eq:8} holds as an equality and as $\lim_{\sigma_b\to 1}$, the upper inequality in \eqref{eq:9} holds as an equality i.e.,
\begin{enumerate}
    \item When $\sigma_b\to 0$, \begin{align}
        \lim_{\sigma_b\to 0} \max_{x\in \NE} B(x) = nb(e^*) \label{eq:epidemic game0}
    \end{align}
    \item When $\sigma_b\to 1$, \begin{align}
        \lim_{\sigma_b\to 1} \max_{x\in \NE} B(x) = nb(e^*) -cne^* + c\alpha_{(A+I)e}(G)e^* \label{eq:epidemic game1}
    \end{align}
\end{enumerate} \label{cor:1}
\end{corollary}

\begin{proof}
\begin{enumerate}
    \item Follows from \eqref{eq:8} and $\lim_{\sigma_b\to 0} b(ne^*) = b(e^*)$ (Lemma \ref{lem:parthe}, point 4).
    \item We will proceed in two steps. First, we will show that  \begin{align}
        \lim_{\sigma_b\to 1} B(x) &= nb(e^*) - cne^* + cE_{(A_G+I_n)e}(x) \label{eq:epidemic game2}
    \end{align} where $E_{w}(x)$ is the $w$-weighted effort. Using this, we prove \eqref{eq:epidemic game1}.
    \begin{enumerate}
     \item For any equilibrium $x$, 
         \begin{align*}
    \lim_{\sigma_b \to 1}B(x) &= \sum_k \lim_{\sigma_b\to 1} b(\mathcal{E}_k(x))\\
    &= \sum_k \bigg(b(e^*) + c\cdot(\mathcal{E}_k(x) -e^*)\bigg) \hspace{1.6em} \\
    &= nb(e^*) -cne^* + cE_{(A_G+I_n)e}(x), 
\end{align*}
where the second equality follows from Lemma \ref{lem:parthe}, point 4.
\item Let $s$ be the specialized equilibrium corresponding to the $(A_G+I_n)e$-weighted maximum independent set. Since $E_{(A_G+I_n)e}(s) = \alpha_{(A+I)e}(G)e^*$, we have $\lim_{\sigma_b\to 1} B(s) = nb(e^*) - cne^* + c\alpha_{(A+I)e}e^*$.
        Since $s\in \NE$, $B(s)$ is weakly worse off than $\max_{x\in \NE} B(x)$. $\max_{x\in \NE} B(x)$ is bounded above by \eqref{eq:9}. Note that this upper bound is equal to $\lim_{\sigma_b\to 1}B(s)$ i.e.,
    \begin{align*}
        B(s) \leq \max_{x\in \NE} B(x) \leq nb(e^*) -cne^* + c\alpha_{(A+I)e}(G)e^* = \lim_{\sigma_b\to 1} B(s)
    \end{align*} Therefore, as $\sigma_b\to 1$, $\lim_{\sigma_b\to 1}\max_{x\in \NE} B(x) = \lim_{\sigma_b\to 1} B(s) = nb(e^*) -cne^* + c\alpha_{(A+I)e}e^*$.
    \end{enumerate}
\end{enumerate} 
\end{proof}

\subsection{Proof of Theorem \ref{th:struc}}
\begin{enumerate}
    \item From Lemma \ref{lem:parthe} point 1, the maximum aggregate effort among all Nash equilibria is given by $\max_{x^*\in \NE} \sum_{k\in V} x^*_k = \alpha(G) \cdot e^*$, where $\alpha(G)$ is the size of the maximum independent set of graph $G$. Since $\alpha(G)$ is a graph property which is unaffected by the signal, $\max_{x^*\in \NE}\mathbb{E}_{\iota, \iota_S}[\sum_k x^*_k]= \alpha\cdot \mathbb{E}_{\iota, \iota_S}[e^*]$.
    \item From LCP \eqref{eq:LCP}, whenever the inverse of $A_G+I_n$ exists, a unique distributed equilibrium exists.
    For this equilibrium, the equilibrium effort profile can be found by solving 
\begin{align*}
    x^* &= (A_G+I_n)^{-1} e^*. \end{align*} 
    Let $M=(A_G+I_n)^{-1}$, with the $kj^{th}$ entry given by $m_{kj}$.
    The effort of individual $k$, and therefore the aggregate effort in this equilibrium is given by
    \begin{align*}
    x^*_k &= e^*\sum_j m_{kj}\\
    \sum_k x^*_k &= e^* \sum_k\sum_j m_{kj}.
\end{align*} Notice that $m(G):=\sum_{k}\sum_{j}m_{kj}$ is a graph property.

The inverse of $(A_G+I_n)$ does not exist when two neighbouring nodes have the same set of neighbours. This situation corresponds to that of multiple distributed equilibria and can be reduced to an equivalent situation of the unique distributed equilibrium by performing the following transformation:

Remove the row and column corresponding to any of the two neighbouring nodes having the same neighbours. On doing this for all such pairs of nodes, the resultant matrix has an inverse and has the same aggregate effort as the original matrix. This aggregate effort would be given by $\text{(some graph property)}\cdot e^*$. 

\item Follows from equation \eqref{eq:eqobj2} and Lemma \ref{lem:2}. 
\item Follows from equation \eqref{eq:eqobj2} and Corollary \ref{cor:1}. Equation \eqref{eq:epidemic game0} and Lemma \ref{lem:2} implies that in the realm of $\sigma_b\to 0$, all equilibria perform the same in terms of the aggregate benefit. The equilibrium maximizing the aggregate benefit when $\sigma_b\to 1$, is the specialized equilibrium corresponding to the $(A_G+I_n)e$-weighted maximum independent set.
\end{enumerate}

\subsection{Proof of Theorem \ref{th:policy}}
From the structural characterization in Theorem \ref{th:struc}, we saw that there were two main reductions of the government objectives: $e^*$ and $b(e^*)$. Recall that the optimal signaling policy depends on the curvature of the government's payoff as a function of the public's belief (cf. \cite{persuasion}, or Section \ref{sec:preliminaries}) i.e., when the government's payoff is strictly concave in public's belief no disclosure is optimal and when it is strictly convex, full disclosure is optimal. When $\mu$ is the belief of individuals of the high infection state, from the best-response condition, the effort $e^*$ satisfies 
\begin{align}
    \mu b'(e^*;h) + (1-\mu) b'(e^*;l) = c. \label{eq:b'}
\end{align}
The curvature of $e^*$ with respect to $\mu$, can be found by differentiating \eqref{eq:b'} with respect to $\mu$
\begin{align}
    {e^*}'(\mu) &= \frac{b'(e^*;l) - b'(e^*;h)}{b''(e^*)} \label{eq:x'}\\
    {e^*}''(\mu) &= \frac{{e^*}'}{-b''(e^*)}\bigg[ b'''(e^*){e^*}' + 2(b''(e^*;h) - b''(e^*;l))\bigg]. \label{eq:x''}
\end{align}

Note that from Assumption \ref{asu:1}, ${e^*}'(\mu)>0$, $\forall \mu$. Thus, $\frac{{e^*}'}{-b''(e^*)}>0$ in \eqref{eq:x''}. Also note, from \eqref{eq:x'}, $\Delta b' := b'(x;l) - b'(x;h) <0$ since $b$ is a concave function. Recall $\mathcal{R}(b;\mu):=2\mathcal{A}(\Delta b) - \mathcal{P}(b)$. Then $e^*(\mu)$ is strictly concave (convex) if the term in the square bracket of \eqref{eq:x''} is negative (positive) which can be expressed in terms of the sign of $\mathcal{R}(b;\mu)$ to give
\begin{align}
    e^*(\mu) =\begin{cases}
			\text{strictly convex in }\mu, & \text{if } \hspace{0.5em} \mathcal{R}(b;\mu) <0 , \hspace{0.5em} \forall \mu\\
            \text{strictly concave in }\mu, & \text{if } \hspace{0.5em} \mathcal{R}(b;\mu) >0, \hspace{0.5em} \forall \mu.
		 \end{cases} \label{eq:x*}
\end{align} This proves Theorem \ref{th:policy} parts 1 and 2.

Similarly, when the government's objectives reduce to $b(e^*)$, the optimal policy can be obtained by studying its curvature with respect to $\mu$ by differentiating and simplifying its expression given by $b(e^*(\mu))=\mu b(e^*;h)+(1-\mu)b(e^*;l)$. Recall $\tilde{\mathcal{R}}(b;\mu):= 2\mathcal{A}(\Delta b)-\mathcal{P}(b)-\mathcal{A}(b)$. Then, \begin{align}
    b(e^*(\mu)) =\begin{cases}
			\text{strictly convex in } \mu, & \text{if } \hspace{0.5em} \tilde{\mathcal{R}}(b;\mu) <0 , \hspace{0.5em} \forall \mu\\
            \text{strictly concave in } \mu, & \text{if } \hspace{0.5em} \tilde{\mathcal{R}}(b;\mu) >0, \hspace{0.5em} \forall \mu.
		 \end{cases} \label{eq:b(x*)}
\end{align} This proves Theorem \ref{th:policy} parts 3 and 4(a).

For part 4(b), when $\sigma_b\to 1$, this objective being a positive linear combination of the $e^*$ and $b(e^*)$. Since $\mathcal{A}(b)>0$, $\tilde{\mathcal{R}}(b;\mu)>0 \implies$ $\mathcal{R}(b,\mu)>0$ and $\mathcal{R}(b;\mu)<0 \implies$ $\tilde{\mathcal{R}}(b,\mu)<0$, $b(e^*)$'s no disclosure conditions and $e^*$'s full disclosure conditions are sufficient for ensuring no disclosure and full disclosure for positive linear combinations of $e^*$ and $b(e^*)$ respectively.

\vspace{1.5em}

\subsubsection{Proof of Corollary \ref{cor:suff_cond}}

$\mathcal{R}(b;\mu)>0$ if and only if  
\begin{align}
    \mu\bigg[ 2\frac{(\Delta b)''}{(\Delta b)'} (\Delta b)'' - (\Delta b)''' \bigg] &> 2\frac{(\Delta b)''}{(\Delta b)'} b''(x;l) - b'''(x;l). \label{eq:Y,Z}
\end{align}
Note $Y(b(\cdot;h), b(\cdot;l)):= 2\frac{(\Delta b)''}{(\Delta b)'} (\Delta b)'' - (\Delta b)''' $ and $Z(b(\cdot;h), b(\cdot;l)) := 2\frac{(\Delta b)''}{(\Delta b)'} b''(x;l) - b'''(x;l)$ are independent of $\mu$. Moreover, if $Z(b(\cdot;h), b(\cdot;l))$ is negative and $Y(b(\cdot;h), b(\cdot;l))$ is positive for all values of $x$, then the above condition holds for all $\mu$, thereby from \eqref{eq:x*}, giving a sufficient condition for full information as the optimal signaling policy. Similar calculations show that the sufficient conditions for no information to be the optimal signaling policy is $Y(b(\cdot;h), b(\cdot;l))<0$ and $Z(b(\cdot;h), b(\cdot;l))>0$ $\forall x$.

\subsection{Proof of Theorem \ref{th:partial policy}}

\begin{figure}[htp]
     \centering
     \begin{subfigure}[b]{0.495\textwidth}
         \centering
         \includegraphics[width=\textwidth]{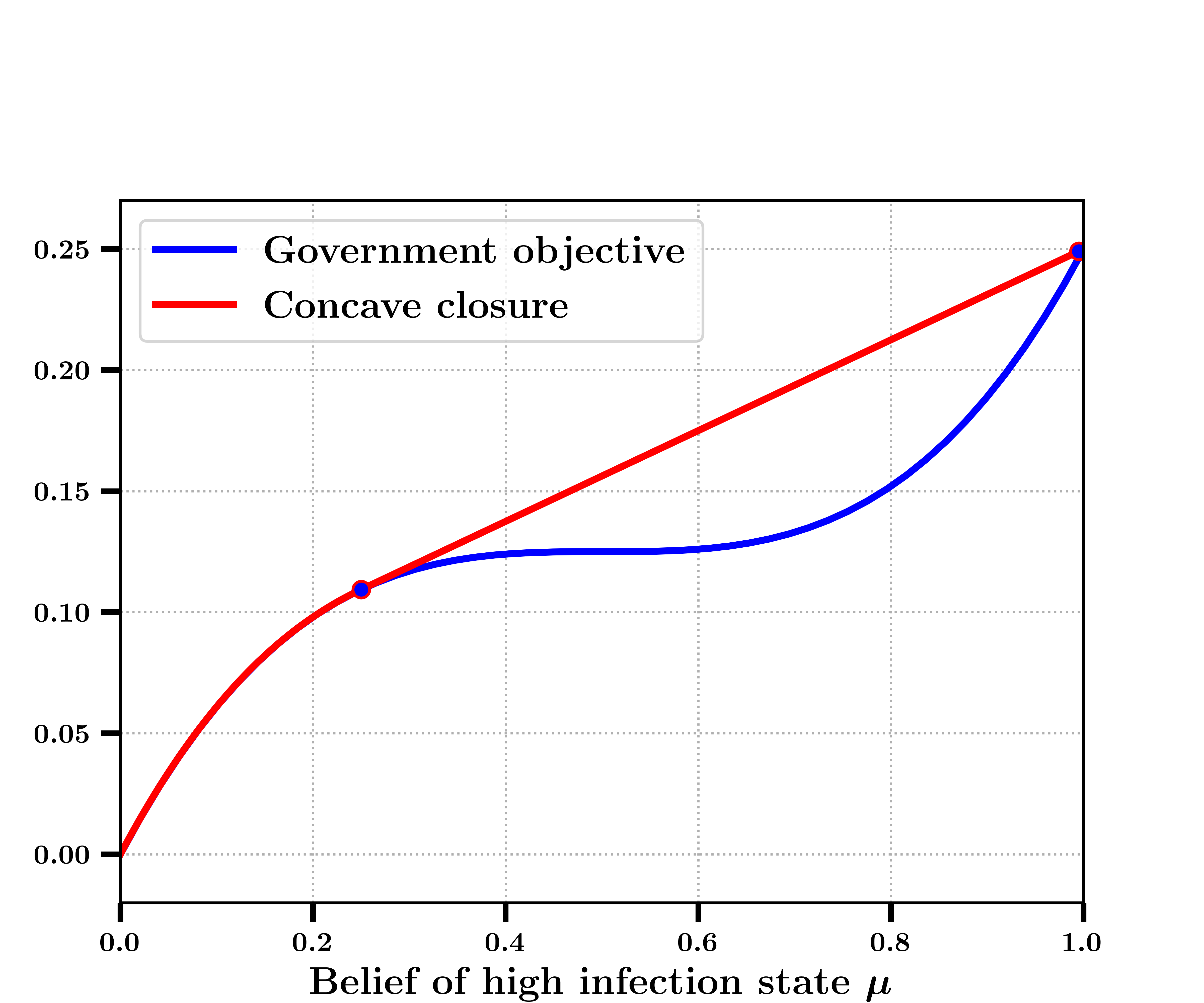}
     \end{subfigure}
     \hfill
     \begin{subfigure}[b]{0.495\textwidth}
         \centering
         \includegraphics[width=\textwidth]{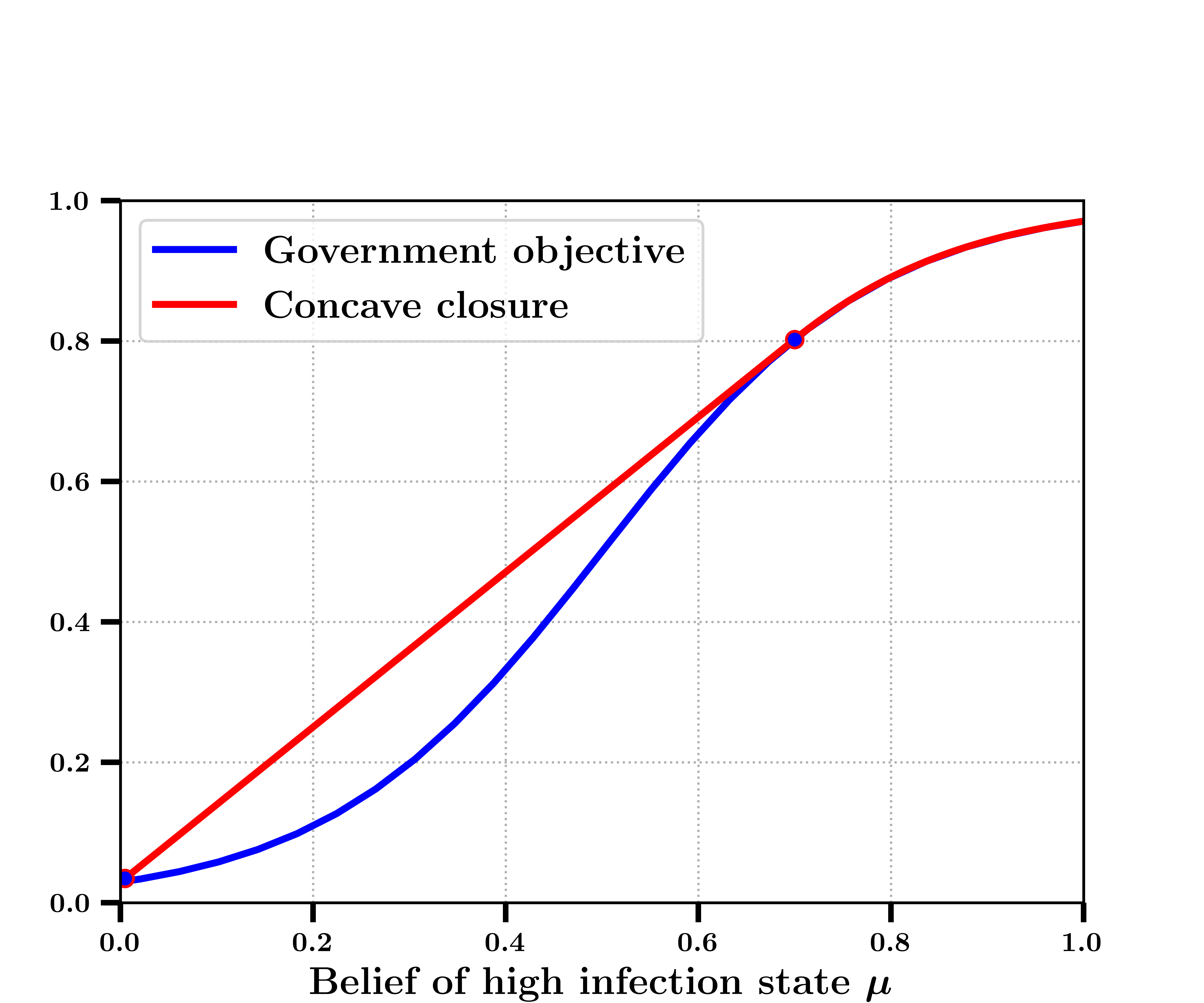}
     \end{subfigure}
        \caption{(L) When government objective is first strictly concave and then strictly convex (R) When government objective is first strictly convex and then strictly concave.}
        \label{fig:thm 3}
\end{figure}


Let $\mu_1:=\sup\{\mu:\forall \tilde{\mu}<\mu, \mathcal{O}(\tilde{\mu}) = \mathcal{C}_{\mathcal{O}}(\tilde{\mu}) \}$, where $\mathcal{O}(\mu)=\mathcal{O}(x^*;\pi)$ is the value of the government objective \eqref{eq:obj1}, \eqref{eq:obj2} when the posterior belief induced is $\mu$, and $\mathcal{C}_{\mathcal{O}}(\mu)$ is the value of concave closure of $\mathcal{O}$ at $\mu$, as defined in \eqref{eq:concave closure}.
Also let $\mu_2:=\inf \{\mu:\forall \tilde{\mu}>\mu, \mathcal{O}(\tilde{\mu}) = \mathcal{C}_{\mathcal{O}}(\tilde{\mu}) \}$. 
According to the concave closure arguments for determining the optimal policy discussed in \cite{persuasion}, the policy which generates beliefs $\mu_1$ and $\mu_2$ through the two signals sent corresponds to the optimal policy. When $p_h:=\pi(\iota_S=h|\iota=h)$, $p_l:=\pi(\iota_S=l|\iota=l)$ and $p_0$ is the prior belief of the high infection state, the posterior distributions of the high infection state on receiving signals $h$ and $l$ correspond to $(p_0p_h)/(p_0p_h+(1-p_0)(1-p_l))$ and $(p_0(1-p_h))/(p_0(1-p_h)+(1-p_0)p_l)$ respectively.

Consider the government's objective as in Figure \ref{fig:thm 3}(L), which is strictly concave at all values of $\mu$ less than some $\mu_d$, and is strictly convex $\forall \mu>\mu_d$. In this case, $\mu_2=1$, since if not (i.e., if $\mu_2<1$), by definition of $\mu_2$ and definition of concave closure, in the interval $(\mu_2,1]$ the government's objective is concave which is a contradiction to the hypothesis assumption that $\mathcal{O}(\mu)$ is strictly convex $\forall \mu\in(\mu_d,1]$. Equating $\mu_2=1$ to the explicitly written out posterior probability correspond to the case $p_l=1$. This is the downplaying policy, as defined in  \eqref{eq:down}. Following \eqref{eq:x*}, $e^*(\mu)$ is concave for $\mu \in [0,\mu_d)$, if $\mathcal{R}(b;\mu)>0\ \forall \mu<\mu_d$ and convex for $\mu \in(\mu_d,1]$  if $\mathcal{R}(b;\mu)<0\ \forall \mu>\mu_d$. Similarly, the results for the objective $b(e^*)$ follow from \eqref{eq:b(x*)}. 

Similarly, when the government's objective is first strictly convex and then strictly concave in $\mu$ (Figure \ref{fig:thm 3})(R), $\mu_1=0$ (since if $\mu_1>0$, then for all $\mu\in[0,\mu_1)$, by definition of $\mu_1$ and definition of concave closure, $\mathcal{O}(\mu)$ is concave, contradicting the hypothesis of the theorem), which when equated to the explicitly written formula for the posterior probability gives $p_h=1$. This is the exaggeration policy, as defined in \eqref{eq:exag}.

\subsubsection{Proof of Proposition \ref{prop:invariance}}

Let $p_0=\mu_0(h)$ be the prior belief of the high infection state.
\begin{enumerate}
    \item For policy $(a,b)\in (p_h,p_l)$, the distribution of posterior beliefs is given by 
\begin{align*}
    \mu_{(a,b)}(r=h|s=h) = \frac{ap_0}{ap_0+(1-b)(1-p_0)} \\
    \mu_{(a,b)}(r=l|s=l) = \frac{(b)(1-p_0)}{(1-a)p_0+b(1-p_0)}.
\end{align*}

For policy $(1-a,1-b)\in (p_h,p_l)$, the distribution of posterior beliefs is given by
\begin{align*}
    \mu_{(1-a,1-b)}(r=h|s=h) &= \frac{(1-a)p_0}{(1-a)p_0+b(1-p_0)} = \mu_{(a,b)}(r=h|s=l)\\
    \mu_{(1-a,1-b)}(r=l|s=l) &= \frac{(1-b)(1-p_0)}{ap_0+(1-b)(1-p_0)} = \mu_{(a,b)}(r=l|s=h).
\end{align*}

\item Under policy $(a,b)$, the expected utility given signal $\iota_S$ is given by \begin{align*}
    \mathbb{E}U_k(x|\iota_S) = \mu_{(a,b)}(\iota=h|\iota_S) b(x;h)+\mu_{(a,b)}(\iota=l|\iota_S) b(x;l)-c\cdot x_k.
\end{align*} Under policy $(1-a,1-b)$, the expected utility given signal $\iota_S^c$ is \begin{align*}
    \mathbb{E}U_k(x|\iota_S^c) &= \mu_{(1-a,1-b)}(\iota=h|\iota_S^c) b(x;h)+\mu_{(1-a,1-b)}(\iota=l|\iota_S^c) b(x;l)-c\cdot x_k\\ &= \mu_{(a,b)}(\iota=h|\iota_S) b(x;h)+\mu_{(a,b)}(\iota=l|\iota_S) b(x;l)-c\cdot x_k.
\end{align*} Thus, the best response under policy $(a,b)$ and signal $\iota_S$ is the same as the best response under policy $(1-a,1-b)$ and signal $\iota_S^c$.
\item 
The expected effort objective under both the policies are given by
\begin{align*}
    \mathbb{E}e^*_{(a,b)} &= p_0[ax^*_{h(a,b)} + (1-a)x^*_{l(a,b)}] + (1-p_0)[(1-b)x^*_{h(a,b)} + bx^*_{l(a,b)}]\\
    \mathbb{E}e^*_{(1-a,1-b)} &= p_0[(1-a)x^*_{h(1-a,1-b)} + a x^*_{l(1-a,1-b)}] + (1-p_0)[bx^*_{h(1-a,1-b)} + (1-b)x^*_{l(1-a,1-b)}],
\end{align*}
where $x^*_{h(\pi)}$ and $x^*_{l(\pi)}$ is the best response unilateral effort when $h$ and $l$ (respectively) are signalled under policy $\pi$.
\end{enumerate}

\end{document}